\documentclass[reqno,12pt,a4paper]{amsart}

\usepackage{mathrsfs}
\usepackage{amssymb}
\usepackage{amsfonts}
\usepackage{latexsym}
\usepackage{amsthm}
\usepackage{mathtools}
\usepackage{etoolbox}
\usepackage{enumitem}
\usepackage{amsmath}
\usepackage{amstext}
\usepackage{graphicx}
\setlist{label={(\roman*)}, wide=0pt, nosep}
\usepackage{titletoc}
\usepackage{pict2e}
\usepackage{graphicx}
\usepackage{tikz}

\numberwithin{equation}{section}

\begin{document}


\renewcommand{\theequation}{\arabic{section}.\arabic{equation}}
\theoremstyle{plain}
\newtheorem{theorem}{\bf Theorem}[section]
\newtheorem{lemma}[theorem]{\bf Lemma}
\newtheorem{corollary}[theorem]{\bf Corollary}
\newtheorem{proposition}[theorem]{\bf Proposition}
\newtheorem{definition}[theorem]{\bf Definition}
\newtheorem*{definition*}{\bf Definition}
\newtheorem*{example}{\bf Example}
\newtheorem*{theorem*}{\bf Theorem}
\newtheorem*{remark}{\bf Remark}

\def\a{\alpha}  \def\cA{{\mathcal A}}     \def\bA{{\bf A}}  \def\mA{{\mathscr A}}
\def\b{\beta}   \def\cB{{\mathcal B}}     \def\bB{{\bf B}}  \def\mB{{\mathscr B}}
\def\g{\gamma}  \def\cC{{\mathcal C}}     \def\bC{{\bf C}}  \def\mC{{\mathscr C}}
\def\G{\Gamma}  \def\cD{{\mathcal D}}     \def\bD{{\bf D}}  \def\mD{{\mathscr D}}
\def\d{\delta}  \def\cE{{\mathcal E}}     \def\bE{{\bf E}}  \def\mE{{\mathscr E}}
\def\D{\Delta}  \def\cF{{\mathcal F}}     \def\bF{{\bf F}}  \def\mF{{\mathscr F}}
\def\c{\chi}    \def\cG{{\mathcal G}}     \def\bG{{\bf G}}  \def\mG{{\mathscr G}}
\def\z{\zeta}   \def\cH{{\mathcal H}}     \def\bH{{\bf H}}  \def\mH{{\mathscr H}}
\def\e{\eta}    \def\cI{{\mathcal I}}     \def\bI{{\bf I}}  \def\mI{{\mathscr I}}
\def\p{\psi}    \def\cJ{{\mathcal J}}     \def\bJ{{\bf J}}  \def\mJ{{\mathscr J}}
\def\vT{\Theta} \def\cK{{\mathcal K}}     \def\bK{{\bf K}}  \def\mK{{\mathscr K}}
\def\k{\kappa}  \def\cL{{\mathcal L}}     \def\bL{{\bf L}}  \def\mL{{\mathscr L}}
\def\l{\lambda} \def\cM{{\mathcal M}}     \def\bM{{\bf M}}  \def\mM{{\mathscr M}}
\def\L{\Lambda} \def\cN{{\mathcal N}}     \def\bN{{\bf N}}  \def\mN{{\mathscr N}}
\def\m{\mu}     \def\cO{{\mathcal O}}     \def\bO{{\bf O}}  \def\mO{{\mathscr O}}
\def\n{\nu}     \def\cP{{\mathcal P}}     \def\bP{{\bf P}}  \def\mP{{\mathscr P}}
\def\r{\varrho} \def\cQ{{\mathcal Q}}     \def\bQ{{\bf Q}}  \def\mQ{{\mathscr Q}}
\def\s{\sigma}  \def\cR{{\mathcal R}}     \def\bR{{\bf R}}  \def\mR{{\mathscr R}}
\def\S{\Sigma}  \def\cS{{\mathcal S}}     
\def\t{\tau}    \def\cT{{\mathcal T}}     \def\bT{{\bf T}}  \def\mT{{\mathscr T}}
\def\f{\phi}    \def\cU{{\mathcal U}}     \def\bU{{\bf U}}  \def\mU{{\mathscr U}}
\def\F{\Phi}    \def\cV{{\mathcal V}}     \def\bV{{\bf V}}  \def\mV{{\mathscr V}}
\def\P{\Psi}    \def\cW{{\mathcal W}}     \def\bW{{\bf W}}  \def\mW{{\mathscr W}}
\def\o{\omega}  \def\cX{{\mathcal X}}     \def\bX{{\bf X}}  \def\mX{{\mathscr X}}
\def\x{\xi}     \def\cY{{\mathcal Y}}     \def\bY{{\bf Y}}  \def\mY{{\mathscr Y}}
\def\X{\Xi}     \def\cZ{{\mathcal Z}}     \def\bZ{{\bf Z}}  \def\mZ{{\mathscr Z}}
\def\O{\Omega}

\newcommand{\mc}{\mathscr {c}}

\newcommand{\gA}{\mathfrak{A}}          \newcommand{\ga}{\mathfrak{a}}
\newcommand{\gB}{\mathfrak{B}}          \newcommand{\gb}{\mathfrak{b}}
\newcommand{\gC}{\mathfrak{C}}          \newcommand{\gc}{\mathfrak{c}}
\newcommand{\gD}{\mathfrak{D}}          \newcommand{\gd}{\mathfrak{d}}
\newcommand{\gE}{\mathfrak{E}}
\newcommand{\gF}{\mathfrak{F}}           \newcommand{\gf}{\mathfrak{f}}
\newcommand{\gG}{\mathfrak{G}}           \newcommand{\Gg}{\mathfrak{g}}
\newcommand{\gH}{\mathfrak{H}}           \newcommand{\gh}{\mathfrak{h}}
\newcommand{\gI}{\mathfrak{I}}           \newcommand{\gi}{\mathfrak{i}}
\newcommand{\gJ}{\mathfrak{J}}           \newcommand{\gj}{\mathfrak{j}}
\newcommand{\gK}{\mathfrak{K}}            \newcommand{\gk}{\mathfrak{k}}
\newcommand{\gL}{\mathfrak{L}}            \newcommand{\gl}{\mathfrak{l}}
\newcommand{\gM}{\mathfrak{M}}            \newcommand{\gm}{\mathfrak{m}}
\newcommand{\gN}{\mathfrak{N}}            \newcommand{\gn}{\mathfrak{n}}
\newcommand{\gO}{\mathfrak{O}}
\newcommand{\gP}{\mathfrak{P}}             \newcommand{\gp}{\mathfrak{p}}
\newcommand{\gQ}{\mathfrak{Q}}             \newcommand{\gq}{\mathfrak{q}}
\newcommand{\gR}{\mathfrak{R}}             \newcommand{\gr}{\mathfrak{r}}
\newcommand{\gS}{\mathfrak{S}}              \newcommand{\gs}{\mathfrak{s}}
\newcommand{\gT}{\mathfrak{T}}             \newcommand{\gt}{\mathfrak{t}}
\newcommand{\gU}{\mathfrak{U}}             \newcommand{\gu}{\mathfrak{u}}
\newcommand{\gV}{\mathfrak{V}}             \newcommand{\gv}{\mathfrak{v}}
\newcommand{\gW}{\mathfrak{W}}             \newcommand{\gw}{\mathfrak{w}}
\newcommand{\gX}{\mathfrak{X}}               \newcommand{\gx}{\mathfrak{x}}
\newcommand{\gY}{\mathfrak{Y}}              \newcommand{\gy}{\mathfrak{y}}
\newcommand{\gZ}{\mathfrak{Z}}             \newcommand{\gz}{\mathfrak{z}}

\def\ve{\varepsilon}   \def\vt{\vartheta}    \def\vp{\varphi}    \def\vk{\varkappa}

\def\A{{\mathbb A}} \def\B{{\mathbb B}} \def\C{{\mathbb C}}
\def\dD{{\mathbb D}} \def\E{{\mathbb E}} \def\dF{{\mathbb F}} \def\dG{{\mathbb G}}
\def\H{{\mathbb H}}\def\I{{\mathbb I}} \def\J{{\mathbb J}} \def\K{{\mathbb K}} \def\dL{{\mathbb L}}
\def\M{{\mathbb M}} \def\N{{\mathbb N}} \def\O{{\mathbb O}} \def\dP{{\mathbb P}} \def\R{{\mathbb R}}
\def\dQ{{\mathbb Q}} \def\S{{\mathbb S}} \def\T{{\mathbb T}} \def\U{{\mathbb U}} \def\V{{\mathbb V}}
\def\W{{\mathbb W}} \def\X{{\mathbb X}} \def\Y{{\mathbb Y}} \def\Z{{\mathbb Z}}

\newcommand{\bpsi}{\boldsymbol{\psi}}
\newcommand{\bS}{\boldsymbol{S}}
\newcommand{\by}{\boldsymbol{y}}
\newcommand{\br}{\boldsymbol{\r}}


\def\la{\leftarrow}              \def\ra{\rightarrow}            \def\Ra{\Rightarrow}
\def\ua{\uparrow}                \def\da{\downarrow}
\def\lra{\leftrightarrow}        \def\Lra{\Leftrightarrow}


\def\lt{\biggl}                  \def\rt{\biggr}
\def\ol{\overline}               \def\wt{\widetilde}
\def\no{\noindent}


\let\ge\geqslant                 \let\le\leqslant
\def\lan{\langle}                \def\ran{\rangle}
\def\/{\over}                    \def\iy{\infty}
\def\sm{\setminus}               \def\es{\emptyset}
\def\ss{\subset}                 \def\ts{\times}
\def\pa{\partial}                \def\os{\oplus}
\def\om{\ominus}                 \def\ev{\equiv}
\def\iint{\int\!\!\!\int}        \def\iintt{\mathop{\int\!\!\int\!\!\dots\!\!\int}\limits}
\def\el2{\ell^{\,2}}             \def\1{1\!\!1}
\def\sh{\sharp}
\def\wh{\widehat}
\def\ds{\dotplus}


\def\all{\mathop{\mathrm{all}}\nolimits}
\def\where{\mathop{\mathrm{where}}\nolimits}
\def\as{\mathop{\mathrm{as}}\nolimits}
\def\Area{\mathop{\mathrm{Area}}\nolimits}
\def\arg{\mathop{\mathrm{arg}}\nolimits}
\def\adj{\mathop{\mathrm{adj}}\nolimits}
\def\const{\mathop{\mathrm{const}}\nolimits}
\def\det{\mathop{\mathrm{det}}\nolimits}
\def\diag{\mathop{\mathrm{diag}}\nolimits}
\def\diam{\mathop{\mathrm{diam}}\nolimits}
\def\dim{\mathop{\mathrm{dim}}\nolimits}
\def\dist{\mathop{\mathrm{dist}}\nolimits}
\def\Im{\mathop{\mathrm{Im}}\nolimits}
\def\Iso{\mathop{\mathrm{Iso}}\nolimits}
\def\Ker{\mathop{\mathrm{Ker}}\nolimits}
\def\Lip{\mathop{\mathrm{Lip}}\nolimits}
\def\rank{\mathop{\mathrm{rank}}\limits}
\def\Ran{\mathop{\mathrm{Ran}}\nolimits}
\def\Re{\mathop{\mathrm{Re}}\nolimits}
\def\Res{\mathop{\mathrm{Res}}\nolimits}
\def\res{\mathop{\mathrm{res}}\limits}
\def\dom{\mathop{\mathrm{dom}}\limits}
\def\sign{\mathop{\mathrm{sign}}\nolimits}
\def\supp{\mathop{\mathrm{supp}}\nolimits}
\def\Tr{\mathop{\mathrm{Tr}}\nolimits}
\def\AC{\mathop{\rm AC}\nolimits}
\def\ind{\mathop{\mathrm{ind}}\nolimits}
\def\BBox{\hspace{1mm}\vrule height6pt width5.5pt depth0pt \hspace{6pt}}


\newcommand\nh[2]{\widehat{#1}\vphantom{#1}^{(#2)}}
\def\dia{\diamond}
\def\Oplus{\bigoplus\nolimits}
\def\qqq{\qquad}
\def\qq{\quad}
\let\ge\geqslant
\let\le\leqslant
\let\geq\geqslant
\let\leq\leqslant

\newcommand{\ca}{\begin{cases}}
\newcommand{\ac}{\end{cases}}
\newcommand{\ma}{\begin{pmatrix}}
\newcommand{\am}{\end{pmatrix}}
\renewcommand{\[}{\begin{equation}}
\renewcommand{\]}{\end{equation}}
\def\bu{\bullet}


\DeclarePairedDelimiterX\set[1]\lbrace\rbrace{
    \def\given{\;\delimsize\vert\;}#1
}
\DeclarePairedDelimiterX\norm[1]\lVert\rVert{
    \ifblank{#1}{\:\cdot\:}{#1}
}
\DeclarePairedDelimiterX\abs[1]\lvert\rvert{
    \ifblank{#1}{\:\cdot\:}{#1}
}
\DeclarePairedDelimiter{\subs}{.}{\rvert}


\title[{Resonances and inverse problems for energy-dependent potentials}]
{Resonances and inverse problems for energy-dependent potentials on the half-line}

\date{\today}

\author[Evgeny Korotyaev]{Evgeny Korotyaev}
\address[Evgeny Korotyaev]{Saint Petersburg State University,
7/9 Universitetskaya Emb., St. Petersburg, 199034, Russia,\
HSE University,
3A Kantemirovskaya Street, St. Petersburg, 194100, Russia,\newline
E-mail address: \texttt{korotyaev@gmail.com, e.korotyaev@spbu.ru}}
    
\author[Andrea Mantile]{Andrea Mantile}
\address[Andrea Mantile]{Laboratoire de Math\'{e}matiques de Reims, UMR9008 CNRS et Universit\'{e} de Reims Champagne-Ardenne,
Moulin de la Housse BP 1039, 51687, Reims, France,\newline
E-mail address: \texttt{andrea.mantile@univ-reims.fr}}

\author[Dmitrii Mokeev]{Dmitrii Mokeev}
\address[Dmitrii Mokeev]{HSE University,
3A Kantemirovskaya Street, St. Petersburg, 194100, Russia,\newline
E-mail address: \texttt{mokeev.ds@yandex.ru,\ dsmokeev@hse.ru}}



\begin{abstract}
We consider Schr\"{o}dinger equations with linearly energy-depending
potentials which are compactly supported on the half-line. We first provide
estimates of the number of eigenvalues and resonances for such complex-valued
potentials under suitable regularity assumptions. Then, we consider a specific
class of energy-dependent Schr\"{o}dinger equations without eigenvalues,
defined with Miura potentials and boundary conditions at the origin. We solve
the inverse resonance problem in this case and describe sets of iso-resonance
potentials and boundary condition parameters. Our strategy consists in
exploiting a correspondance between Schr\"{o}dinger and Dirac equations on the
half-line. As a byproduct, we describe similar sets for Dirac operators and
show that the scattering problem for Schr\"{o}dinger equation or Dirac operator with
an arbitrary boundary condition can be reduced to the scattering problem with
the Dirichlet boundary condition.
\end{abstract}

\maketitle

\section{Introduction} \label{intro}

We consider a Schr\"{o}dinger equation on the half-line%

\[ \label{intro:eq:equation}
    -y''(x)  +V(x, k)  y(x) = k^{2} y(x),\qq (x,k) \in \R_{+} \times \C,
\]
where $k\mapsto V\left(  \cdot,k\right)  $ is a linear function
\[ \label{intro:eq:equation_1}
    V(x,k) := q(x)+2kp(x)
\]
which coefficients $q,p$ have compact support in $\mathbb{R}_{+}$. When
suitable boundary conditions are assigned at the origin,
(\ref{intro:eq:equation}) corresponds to an eigenvalue/resonance problem for
an energy-dependent perturbation of the Laplacian. We will address this problem
for different classes of potentials and boundary conditions. The equation
(\ref{intro:eq:equation}) has many applications; we refer to e.g.
\cite{J76, K07b} and references therein. In particular, if $q=-p^{2}$, we
obtain the massless Klein-Gordon equation (see e.g. \cite{N83}).

In the following, the standard notation $W^{s,p}\left(  {\mathbb{R}}%
_{+}\right)  $ refers to the Sobolev space of order $s\in{\mathbb{R}}$ in
$L^{p}({\mathbb{R}}_{+},\mathbb{C})$, $p\geq1$. When: $p=2$, this is an
Hilbert space denoted as $H^{s}\left(  {\mathbb{R}}_{+}\right)  $. Its dual
identifies with: $\left(  H^{s}\left(  {\mathbb{R}}_{+}\right)  \right)
^{\prime}=\widetilde{H}^{-s}\left(  {\mathbb{R}}_{+}\right)  $, where, for any
$s\in\mathbb{R}$, $\widetilde{H}^{s}\left(  {\mathbb{R}}_{+}\right)  $ is the
closure of $\mathcal{C}_{0}^{\infty}\left(  {\mathbb{R}}_{+}\right)  $ in
$H^{s}\left(  {\mathbb{R}}\right)  $.

\subsection{The potential class $\mathcal{P}$: Eigenvalues and resonances
estimates}

Let $\gamma>0$ and define the class of potential perturbations
$$
    \mathcal{P}:= \set{(p,q) \in W^{1,1}(\mathbb{R}_{+}) \times L^{1}(\mathbb{R}_{+})
    \given \max \supp (\abs{p} + \abs{q}) = \g}.
$$
We consider the problem (\ref{intro:eq:equation}), (\ref{intro:eq:equation_1})
for $\left(  p,q\right)  \in\mathcal{P}$ and the solution $y$ fulfilling the
Dirichlet boundary condition $y(0)=0$. 
Let introduce the Jost solution $y_{+}(x,k)$ of
equation (\ref{intro:eq:equation}) under the condition
\begin{equation}
y_{+}(x,k)=e^{ikx},\quad x\geq\gamma. \label{intro:eq:jostsol_def}%
\end{equation}
In Section \ref{jost}, adapting classical results, we show that the problem
(\ref{intro:eq:equation})-(\ref{intro:eq:jostsol_def}) is well posed and for each $x\geq0$ the function
$k\mapsto y_{+}(x,k)$ is entire (see Lemma \ref{jost:lm:jost_solution}). The
corresponding Jost function is defined by
\[
\psi(k)=y_{+}(0,k),\quad k\in{\mathbb{C}}.
\]
If $\psi(k)=0$ for some $k\in{\mathbb{C}}$, then the Jost solution
$y_{+}(\cdot,k)$ satisfies the Dirichlet boundary condition. By
(\ref{intro:eq:jostsol_def}), one has: $y_{+}(\cdot,k)\in L^{2}({\mathbb{R}%
}_{+})$ if $k\in{\mathbb{C}}_{+}$. Thus, we call the zeros of the Jost
function $\psi$ in ${\mathbb{C}}_{+}$ the eigenvalues of the equation
(\ref{intro:eq:equation}) with the Dirichlet boundary condition. In this
framework, the algebraic multiplicity of an eigenvalue $k\in{\mathbb{C}}_{+}$
is the multiplicity of the zero of the Jost function $\psi$ at the point $k$.
It is easy to check that the function $\psi$ has no zeros on $\R \sm \{0\}$, 
but the origin can be an eigenvalue and it needs additional
considerations. There are also zeros of the Jost function $\psi$ in
${\mathbb{C}}_{-}$ and for such points the solution $y_{+}(\cdot,k)$ is an
exponentially increasing function. We call the zeros of $\psi$ in
${\mathbb{C}}_{-}$ the resonances of equation (\ref{intro:eq:equation}) with
the Dirichlet boundary condition. Their multiplicity is the multiplicity of
the zero of the Jost function $\psi$.

\begin{remark}
    In this section we consider the eigenvalue problem for equation 
    (\ref{intro:eq:equation}), (\ref{intro:eq:equation_1})
    with the Dirichlet boundary conditions at the origin. 
    This problem corresponds to the eigenvalues problem for
    quadratic operator pencils $T(k)$ acting as
    $$
        T(k) y = -y'' + V(\cdot, k) y - k^2 y,\qq
        \dom (T(k)) = H^2(\R_+) \cap H^1_0(\R_+).
    $$

    Let us also notice that the spectral parameter of this model
    is $k^{2}$. Nevertheless, the use of its square root is a common and
    natural choice in spectral analysis. Our definition of eigenvalues and
    resonances should be understood in this particular sense.
\end{remark}

Below we show that the Jost function $\psi(k)$ tends to a non-zero finite limit as
$k \to \iy$ in $\ol \C_+$. Since $\psi$ is entire it follows that $\psi$ has
finitely many eigenvalues counted with multiplicities in $\ol \C_+$. Moreover, if
$\psi(0)\neq0$ and $q$ is real-valued, then the number of eigenvalues counted
with multiplicities of equation (\ref{intro:eq:equation}) with the Dirichlet
boundary condition equals the number of eigenvalues of the Schr\"{o}dinger
equation
\[ \label{intro:eq:schrodinger_equation}%
    -y^{\prime\prime}\left(  x\right)  +q\left(  x\right)  y\left(  x\right)
    =k^{2}y\left(  x\right),\qq y(0)=0,\qq x\geq0,
\]
see Proposition C.1 in \cite{K07a}. Note also that an eigenvalue $k$ of
equation (\ref{intro:eq:schrodinger_equation}) is simple and has zero real
part. In contrast, the eigenvalues of equation (\ref{intro:eq:equation}) with
the Dirichlet boundary condition can have the multiplicity greater than one
and a nontrivial real part (see e.g. \cite{J72}).

Our first results are devoted to bounds for eigenvalues and resonances of
equation (\ref{intro:eq:equation}). We use the norms
\[ \label{intro:eq:norms}%
    \norm{f}_{1} = \int_{\R_+}\abs{f(x)} dx,\qq 
    \norm{f}_{2} = \left(\int_{R_{+}}\abs{f(x)}^{2} dx\right)^{1/2},\qq
    \norm{f}_{w} = \int_{\R_{+}} x \abs{f(x)} dx,
\]
and the constants
\[ \label{intro:eq:constants}%
    C_{1} = \left\Vert q\right\Vert _{1}+\left\vert p\left(  \gamma\right)
    \right\vert +\left\Vert p^{\prime}\right\Vert _{1}+\left\Vert p\right\Vert
    _{2}^{2},\qq
    C_{2} = \max\left(  \left\Vert q\right\Vert _{w}+2\left\Vert
    p\right\Vert _{w}\,,\left\Vert q\right\Vert _{1}+2\left\Vert p\right\Vert
    _{1}\right).
\]
For real $(p,q)\in{\mathcal{P}}$ we also introduce the notation
\[
    p_{+} = \max_{x \in [0, \gamma]} p(x),\qq
    p_{-} = \min_{x \in [0, \gamma]} p(x).
\]

\begin{theorem}
\label{intro:thm:eigenvalues_and_resonances_bounds}Let $(p,q)\in\mathcal{P}$.
Let $k_{o}\in{\mathbb{C}}_{+}$ be an eigenvalue and let $r_{o}\in{\mathbb{C}%
}_{-}$ be a resonance of equation (\ref{intro:eq:equation}) with the Dirichlet
boundary condition. Then they satisfy%
$$
    \left\vert k_{o}\right\vert \leq C_{1}\left(1+e^{2\left\Vert p\right\Vert _{1}
    }\right),\qq \left\vert r_{o}\right\vert \leq C_{1}e^{C_{2}+2\gamma\left\vert
    \Im r_{o}\right\vert }.
$$
Moreover, if $\left\vert k_{o}\right\vert \geq1$, then%
\begin{equation}
0\leq\Im k_{o}\leq\frac{\gamma}{2}\left(\left\Vert q\right\Vert
_{2}+2\left\Vert p\right\Vert _{2}\right)^{4}\,.
\label{intro:eq:imaginary_part_eigenvalue_bound}%
\end{equation}
If $(p,q)\in{\mathcal{P}}$ is real-valued, then we have
\begin{equation}
p_{-}\leq \Re k_{o}\leq p_{+}\,.
\label{intro:eq:real_part_eigenvalue_bound}%
\end{equation}

\end{theorem}

\begin{remark}
    There are many results about bounds of eigenvalues of
    the Schr\"{o}dinger operator with complex-valued potentials given by
    (\ref{intro:eq:schrodinger_equation}). For example, each eigenvalue $k_{o}$
    of this operator satisfies $|k_{o}|\leq\Vert q\Vert$, see
    \cite{FLS11}.
\end{remark}

Now we estimate the number of resonances and eigenvalues in a large disk. Let
$N(r,f)$ be the number of zeros of a function $f$ counted with multiplicities
in the disc of a radius $r$ centered at the origin.

\begin{theorem} \label{intro:thm:resonances_estimate}
    Let $(p,q)\in{\mathcal{P}}$. Then for any $r\geq0$ we have
    \[
        {\mathcal{N}}(r,\psi)\leq\frac{4\gamma}{\pi\log2}r+C_{3},
    \]
    where
    \begin{equation}
        C_{3}=3\gamma C_{1}e^{2\Vert p\Vert_{1}}+3. \label{intro:eq:constant}%
    \end{equation}
    Moreover, if $(p,q)$ is real-valued, then the number of eigenvalues
    ${\mathcal{N}}_{+}$ of equation (\ref{intro:eq:equation}) with the Dirichlet
    boundary condition satisfies
    \[
        {\mathcal{N}}_{+}\leq{\mathcal{N}}(R,\psi),\qq 
        \text{where $R:=\max\left(
        1,\max(p_{+},p_{-})+\frac{\gamma}{2}\left(\Vert q\Vert_{2}+2\Vert p\Vert_{2}%
        \right)^{4}\right)$}.
    \]
\end{theorem}

\begin{remark}
    An estimate on the number of eigenvalues of
    Schr\"{o}dinger operator on the half-line with a complex-valued potential was
    obtained in \cite{FLS16}. The case of compactly supported potentials was
    considered in \cite{K20}.
\end{remark}

\subsection{The potential class $\mathcal{Q}$: Inverse
resonances problems} \label{invres}

Let $\gamma>0$ and define the class of potential perturbations
$$
    \mathcal{Q} :=
    \set*{(p,q) \in L^{2}(\mathbb{R}_{+}) \times \left(H^{1}\left(\mathbb{R}_{+}\right)\right)^{\prime}
    \given q=u^{\prime}+u^{2},\, u\in L^{2}(\mathbb{R}_{+}),\, \max \supp \left(  \left\vert p\right\vert +\left\vert u\right\vert \right)
    = \g}.
$$
Since the potential $q$ has the form
$q=u^{\prime}+u^{2}$ for some $u\in L^{2}(0,\gamma)$, it is a
distribution and it is called a Miura potential. It is known that there exists
a unique $u\in L^{2}(0,\gamma)$ such that $q=u^{\prime}+u^{2}$, see
\cite{KPST05}. One of the usual way to introduce a solution of differential
equation with such distributional potentials is to use the
\textit{quasi-derivative} $y^{[1]}=y^{\prime}-uy$ (see, e.g., \cite{W87}). In
this spirit, we say that $y$ is a solution of equation
(\ref{intro:eq:equation}), (\ref{intro:eq:equation_1}) if $y^{[1]}$ is locally
absolutely continuous and the following identity holds true
$$
-\left(  y^{[1]}\right)  ^{\prime}-u\,y^{[1]}+2kp\,y=k^{2}y\,.
$$

We consider the problem (\ref{intro:eq:equation}), (\ref{intro:eq:equation_1})
with $\left(  p,q\right)  \in\mathcal{Q}$ and the boundary condition%
\[ \label{intro:eq:boundary_condition}%
    y^{[1]}(0)\sin\alpha+ky(0)\cos\alpha=0
\]
parametrized by $\alpha\in[0,\pi)$.
If $\alpha=0$, then we obtain the Dirichlet boundary condition. Note that if
$q\in L^{2}({\mathbb{R}}_{+})$, then $u\in C({\mathbb{R}}_{+})$ and the
boundary condition in (\ref{intro:eq:boundary_condition}) recasts as%
\[
y^{\prime}(0)\sin\alpha+ky(0)\cos\alpha=-u(0)\sin\alpha.
\]
\begin{remark}
    The eigenvalues problem for equation (\ref{intro:eq:equation}), (\ref{intro:eq:equation_1}) with
    boundary condition (\ref{intro:eq:boundary_condition}) corresponds to the eigenvalues problem
    for the quadratic operator pencil $T(k)$ considered in \cite{HM20} and acting as
    $$
        T(k) y = -\left(  y^{[1]}\right)  ^{\prime}-u\,y^{[1]}+2kp\,y - k^{2}y.
    $$
\end{remark}

As above, we introduce the Jost solution $y_{+}(x,k)$ of equation
(\ref{intro:eq:equation}) satisfying $y_{+}(x,k)=e^{ikx}$ when $x\geq\gamma$.
Due Lemma \ref{trans:lm:jost_solution}, $y_{+}(\cdot,k)$ exists for any
$k\in{\mathbb{C}}$ and, for each $x\geq0$, the function $k\mapsto y_{+}(x,k)$
is entire. Let us fix $\alpha\in[0,\pi)$ and introduce the Jost
function
$$
    \psi_{\alpha}(k)=y_{+}^{[1]}(0,k)\sin\alpha+ky_{+}(0,k)\cos\alpha,\qq k \in {\mathbb{C}}.
$$
This is an entire function and its zeros in $\overline{\mathbb{C}}_{+}$ are
the eigenvalues of the problem (\ref{intro:eq:equation}), (\ref{intro:eq:equation_1}), (\ref{intro:eq:boundary_condition}), while the
zeros in ${\mathbb{C}}_{-}$ are the resonances. Their multiplicities are the
multiplicities of the corresponding zeros of $\psi_{\alpha}$. When
$q=u^{\prime}+u^{2}$ for some $u\in L^{2}({\mathbb{R}}_{+})\cap L^{1}%
({\mathbb{R}}_{+})$, it has been shown in \cite{HM20} that the problem
(\ref{intro:eq:equation}), (\ref{intro:eq:equation_1}),
(\ref{intro:eq:boundary_condition}) has no eigenvalues. Hence, for
$(p,q)\in{\mathcal{Q}}$, this problem has no eigenvalues and the Jost function
$\psi_{\alpha}$ has no zeros in $\overline{\mathbb{C}}_{+}$. 

For any
$\alpha\in[0,\pi)$, we introduce the scattering matrix
$$
    S_{\alpha}(k)=\frac{\overline{\psi_{\alpha}}(k)}{\psi_{\alpha}(k)},\qq
    k\in{\mathbb{R}}.
$$
Since $\psi_{\alpha}$ is entire, it follows that $S_{\alpha}$ admits a
meromorphic continuation onto ${\mathbb{C}}$ and its poles are the resonances
of our problem.

For any continuous function $f:{\mathbb{R}}\rightarrow\mathbb{C}%
\backslash\{0\}$ with finite non-zero limits at infinity, we introduce the
winding number
\[
\text{ind}\,f=\lim_{x\rightarrow+\infty}\ln f(x)-\lim_{x\rightarrow-\infty}\ln
f(x)\,,
\]
and define the following classes of the scattering matrices.\smallskip

\begin{definition}
    Let
    \[
        {\mathcal{S}}=\bigcup_{\beta\in[0,\pi)}{\mathcal{S}}_{\beta}\,,
    \]
    where, for each $\beta\in[0,\pi)$, $S_{\beta}$ is
    the set of all meromorphic functions $S$ such that:
    \begin{enumerate}
        \item $\left\vert S(k)\right\vert =1$ for all $k\in\mathbb{R}$;
        \item $\ind S = 0$;
        \item $S$ admit the following representation
        $$
            S(k) = e^{2i\beta} + \int_{-\gamma}^{+\infty} F(s)e^{2iks} ds,\qq k \in {\mathbb{C}},
        $$
        for some $F \in L^{2}(\mathbb{R},\mathbb{C})\cap L^{1}(\mathbb{R},\mathbb{C})$
        such that: $\min \supp F = -\gamma$.
    \end{enumerate}
\end{definition}

Our first result is a solution of inverse scattering problem for equation
(\ref{intro:eq:equation}).

\begin{theorem} \label{intro:thm:scattering_inverse_problem}
    The mapping $(p,q,\alpha)\mapsto
    S_{\alpha}(\cdot,p,q)$ from ${\mathcal{Q}}\times[0,\pi)$ into
    ${\mathcal{S}}$ is a bijection. Moreover, we have $S_{\alpha}(\cdot
    ,p,q)\in{\mathcal{S}}_{\beta}$, where%
    $$
        \beta=-\alpha-\varphi(0) \mod \pi,\qq \varphi(0)=\int_{0}^{+\infty}p(x) dx.
    $$
\end{theorem}

\begin{remark}
    Note that the scattering matrices for equation
    (\ref{intro:eq:equation}) with the same boundary conditions may have different
    asymptotics at infinity. That is for some $S\in S_{\beta}$ and
    $\beta\in[0,\pi)$ we cannot simply determine the corresponding
    boundary condition since it depends on $\varphi(0)$. Below we show,
    that resonances are devoid of this disadvantage.
\end{remark}

Now, we describe the inverse resonances problem for equation
(\ref{intro:eq:equation}) with boundary condition
(\ref{intro:eq:boundary_condition}).

\begin{definition}
    Let $\mathbf{s}$ be a set of all sequences
    $(r_{n})_{n\geq1}$ of points in $\mathbb{C}_{-}$ arranged by
    \begin{equation}%
        \begin{cases}
            0\leq\left\vert r_{1}\right\vert \leq\left\vert r_{2}\right\vert \leq\ldots,\\
            \Re r_{n} \leq \Re r_{n+1},\quad\text{if }\left\vert
            r_{n}\right\vert \leq\left\vert r_{n+1}\right\vert \,.
        \end{cases}
        \label{intro:eq:arrangement}%
    \end{equation}
\end{definition}

For $\alpha\in[0,\pi)$, we introduce the mapping $\varrho^{\alpha}: \cQ \to \mathbf{s}$ as
$$
    \varrho^{\alpha}(p,q):= \text{sequence of zeros of $\psi_{\alpha}\left(
    \cdot,p,q\right)$ in $\mathbb{C}_{-}$ arranged by
    (\ref{intro:eq:arrangement}).}
$$

According to this definition, $\varrho^{\alpha}(p,q)$ is a sequence of
resonances of equation (\ref{intro:eq:equation}) with the boundary condition
(\ref{intro:eq:boundary_condition}). In Section \ref{trans}, we show that, for
potentials from ${\mathcal{Q}}$, the problem (\ref{intro:eq:equation}%
), (\ref{intro:eq:equation_1}), (\ref{intro:eq:boundary_condition}) recasts as
a generalized eigenvalue problem for a suitable Dirac operator on the
half-line with a compactly supported potential depending on $(p,q)\in
\mathcal{Q}$ and boundary condition, induced from
(\ref{intro:eq:boundary_condition}) and parametrized by 
$\alpha \in [0,\pi)  $. The resonances of Dirac operators were considered in
\cite{K14, IK14a, IK14b, KM21} and all properties of resonances obtained in
these papers can be carried over to our case. In Section \ref{dirac} we focus
on the inverse resonance problem for a class of Dirac operators on the
half-line consistent with those obtained by transforming our Schr\"{o}dinger
problem. In this framework, we provide characterizftion of iso-resonance families and
solve the inverse-resonance problem in the self-adjoint case. Then,
building on the transformation linking the Dirac and the Schr\"{o}dinger
problems, we solve the inverse resonances Schr\"{o}dinger problem by
exploiting the solution of the inverse resonances problem for Dirac operators
with compactly supported potentials under the condition that the equation
(\ref{intro:eq:equation}) has no eigenvalues.

We introduce the set
$$
    \mathcal{X} := \set*{f \in L^{2}(\mathbb{R}_{+},\mathbb{C}) \given \max \supp f = \gamma}.
$$
\begin{definition}
    $\mathcal{R}$ is a set of all sequences $(r_{n})_{n\geq1}$ of zeros,
        counted with multiplicity and arranged by (\ref{intro:eq:arrangement}), of
        some entire functions $\psi$, which has no zeros
        in $\overline{\mathbb{C}}_{+}$ and admit representation
        \begin{equation}
        \psi(k)=1+\int_{0}^{\gamma}h(s)e^{2iks}ds\,,\quad k\in{\mathbb{C}},\quad
        h\in{\mathcal{X}}.\label{psi_id}%
        \end{equation}
\end{definition}
It is worth mentioning that the function $\psi$ from this definition is a Jost
function of a Dirac operator with compactly supported potential and Dirichlet
boundary condition. Hence, $\mathcal{R}$ corresponds to the class of
resonances of such operators.\smallskip

\begin{theorem} \label{intro:thm:resonances_inverse_problem}
    \begin{enumerate}
        \item For any $\a \in [0,\pi)$, the mapping $\r^{\a}: \cQ \to \cR$ is a bijection.
        \item For any $(\z,\b) \in \cR \ts [0,\pi)$ there exists a unique
        $(p,q,\a) \in \cQ \ts [0,\pi)$ such that
        $\r^{\a}(p,q) = \z$ and $S_{\a}(\cdot,p,q) \in \cS_{\b}$.
    \end{enumerate}
\end{theorem}

\begin{remark}
    For any boundary condition
    (\ref{intro:eq:boundary_condition}) equation (\ref{intro:eq:equation}) has the
    same class of resonances and due to $(i)$ for any boundary
    condition (\ref{intro:eq:boundary_condition}) the sequence of resonances
    uniquely determines a potential of equation (\ref{intro:eq:equation}).
\end{remark}

It follows from Theorem \ref{intro:thm:resonances_inverse_problem} that
equation (\ref{intro:eq:equation}) with potentials from ${\mathcal{Q}}$ has
exactly the same set of resonances as a Dirac operator with compactly
supported potentials and Dirichlet boundary condition. Building on results
from \cite{K14, IK14a, IK14b, KM21}, we obtain the forbidden domain for
resonances of equation (\ref{intro:eq:equation}) with potentials from
${\mathcal{Q}}$.

\begin{theorem}
\label{intro:thm:resonances_forbidden_domain} Let $r_{o}\in{\mathbb{C}}_{-}$
be a resonance of equation (\ref{intro:eq:equation}) with boundary condition
(\ref{intro:eq:boundary_condition}) for some $(p,q,\alpha)\in{\mathcal{Q}%
}\times[0,\pi)$ and let $\varepsilon>0$. Then there exists a constant
$C=C(\varepsilon,p,q,\alpha)\geq0$ such that the following inequality holds
true:
\begin{equation}
    2\gamma\Im r_{o}\leq\ln\left(  \varepsilon+\frac{C}{|r_{o}|}\right).
    \label{intro:eq:forbidden_domain}%
\end{equation}
In particular, for any $A>0$, there are only finitely many resonances in the
strip
\begin{equation}
    \set*{z \in {\mathbb{C}} \given 0 > \Im z > -A}.
    \label{intro:eq:forbidden_domain_strip}%
\end{equation}

\end{theorem}

Recall that if $p=0$, then (\ref{intro:eq:equation}) is a Schr\"{o}dinger
equation with a potential which does not depend on the energy. For such
equation the inverse resonances problem was solved in \cite{K04a} in case
$q\in L^{1}(0,\gamma)$. Using Theorem
\ref{intro:thm:resonances_inverse_problem}, we solve this problem for
$q\in{\mathcal{Q}}_{o}$, where
\[
    \mathcal{Q}_{o} :=\set*{ q \given \left(0, q\right) \in{\mathcal{Q}}} .
\]
In this case, the set of resonances ${\mathcal{R}}_{o}$ consists of all
sequences from ${\mathcal{R}}$ which are symmetric with respect to the
imaginary line. In order to rigorously define this set, we introduce a reflection
operator $R:\mathbf{s} \rightarrow \mathbf{s}$ such that:
$R(r)$ is a rearrangement of the sequence $r = (-\overline{r_{n}})_{n\geq1}$
satisfying (\ref{intro:eq:arrangement}). 

\begin{definition}
    We denote with
    $\cR_{o}$ the subset of all sequences $r\in \cR$ such that $R r = r$.
\end{definition}

Recall that $\varrho^{o}(0,q)$ is a sequence of resonances of a Schr\"odinger
equation with a potential $q$, which does not depend on the energy. In the
following theorem we solve the inverse resonances problem for such equation.

\begin{theorem}
\label{intro:thm:resonances_inverse_problem_schrodinger_operator} The mapping
$q \mapsto\varrho^{o}(0,q)$ from ${\mathcal{Q}}_{o}$ into ${\mathcal{R}}_{o}$
is a bijection.
\end{theorem}

Note that it follows from Theorem \ref{intro:thm:resonances_inverse_problem},
that one sequence from ${\mathcal{R}}$ can be a sequence of resonances for
various potentials and boundary conditions. Now we describe the potentials
which have the same sequence of resonances with different boundary condition
parameters. For any $(p_{o},q_{o},\alpha_{o})\in{\mathcal{Q}}\times
[0,\pi)$, we introduce the set of ``isoresonances''  potentials and
boundary condition parameters Iso$(p_{o},q_{o},\alpha_{o})$ as follows
$$
    \Iso(p_{o},q_{o},\alpha_{o}) = 
    \set*{(p,q,\alpha) \in \mathcal{Q} \times [0,\pi) \given 
    \varrho^{\alpha}(p,q) = \varrho^{\alpha_{o}}(p_{o},q_{o})}.
$$
We also introduce the mapping ${\mathcal{T}}:{\mathcal{Q}}\rightarrow
{\mathcal{X}}$ by
\begin{equation}
{\mathcal{T}}:(p,q)\mapsto(-u+ip)e^{-2i\varphi}\,,\quad{q=u}^{\prime}%
+u^{2}\,,\quad\varphi(x)=\int_{x}^{+\infty}p(t)dt,\quad x\in{\mathbb{R}}%
_{+}\,. \label{intro:eq:potential_mapping}%
\end{equation}

\begin{remark}
    In Section \ref{trans}, we will use the mapping
    $\mathcal{T}$ to transform a potential of equation
    (\ref{intro:eq:equation}) into potential of a corresponding Dirac equation and
    we will show that this mapping, actually, a bijection. Note also that the
    mapping $\mathcal{T}$ is a composition of two non-linear mappings:
    $q\mapsto u$, where ${q=u}^{\prime}+u^{2}$, and $(p,u)\mapsto
    (-u+ip)e^{-2i\varphi}$, where $\varphi(x)=\int_{x}^{+\infty}%
    p(t)dt$, $x\in \R_{+}$. The first mapping is the inverse mapping
    of the Miura mapping, which is widely studied (see, e.g.,
    \cite{K02,K03,KPST05,IK17}).
\end{remark}

\begin{theorem}
\label{intro:thm:isoresonance_potentials} Let $(p_{o},q_{o},\alpha_{o}%
)\in{\mathcal{Q}}\times[0,\pi)$ and let $v_{o}={\mathcal{T}}(p_{o}%
,q_{o})$. Then we have
\begin{equation}
    \Iso(p_{o},q_{o},\alpha_{o})=
    \set*{(p_{\delta},q_{\delta},\alpha_{\delta}) \given 
    \delta\in[0,\pi)},
    \label{intro:eq:isoresonances_potentials_set}%
\end{equation}
where
\begin{equation}
\begin{aligned} \alpha_{\d} &= \alpha_o + \vartheta_o(0) - \vartheta_{\d}(0) \mod \pi,\\ p_{\d} &= \mathop{\mathrm{Im}}\nolimits v_o \cos 2 \vartheta_{\d} + \mathop{\mathrm{Re}}\nolimits v_o \sin 2 \vartheta_{\d},\\ u_{\d} &= -\mathop{\mathrm{Re}}\nolimits v_o \cos 2 \vartheta_{\d} + \mathop{\mathrm{Im}}\nolimits v_o \sin 2 \vartheta_{\d},\\ q_{\d} &= u_{\d}' + u_{\d}^2, \end{aligned} \label{intro:eq:isoresonances_potentials}%
\end{equation}
and $\vartheta_{\delta}$ is a unique solution of the initial value problem
\begin{equation}%
\begin{cases}
\vartheta_{\delta}^{\prime}=-\Im v_{o}\cos2\vartheta_{\delta
}-\Re v_{o}\sin2\vartheta_{\delta},\\
\vartheta_{\delta}(\gamma)=\delta.
\end{cases}
\label{intro:eq:isoresonances_potentials_ivp}%
\end{equation}
Moreover, $\vartheta_{\delta}(0)$ is a strictly increasing function of
$\delta$ and $\vartheta_{\pi}(0)=\vartheta_{o}(0)+\pi$.
\end{theorem}

Now, we describe the scattering matrices and potentials of the corresponding
Dirac operators for the set of ``isoresonances'' potentials and boundary
condition parameters.

\begin{theorem}
\label{intro:thm:isoresonances_scattering_matrices} Let $(p_{o},q_{o}%
,\alpha_{o})\in{\mathcal{Q}}\times[0,\pi)$ and let Iso$(p_{o}%
,q_{o},\alpha_{o})$ be given by (\ref{intro:eq:isoresonances_potentials_set})
and (\ref{intro:eq:isoresonances_potentials}). Then for any $\delta\in
[0,\pi)$ the following identities hold true:
\begin{equation}
{\mathcal{T}}(p_{\delta},q_{\delta})=e^{2i\delta}{\mathcal{T}}(p_{o}%
,q_{o})\,,\quad S_{\alpha_{\delta}}(\cdot,p_{\delta},q_{\delta})=e^{2i\delta
}S_{\alpha_{o}}(\cdot,p_{o},q_{o})\,.
\label{intro:eq:isoresonances_scattering_matrices}%
\end{equation}

\end{theorem}

\begin{remark}
    In Theorem \ref{dirac:thm:isoresonances_potentials}, we
    describe the set of isoresonances potentials for Dirac operator with compactly
    supported potentials and we also describe the scattering matrices for such
    potentials. The first identity in
    (\ref{intro:eq:isoresonances_scattering_matrices}) implies that $\mathcal{T}$ maps $\Iso (p_{o},q_{o},\alpha_{o})$
    onto the the set of
    isoresonances potentials for Dirac operator. And it follows from the second
    identity in (\ref{intro:eq:isoresonances_scattering_matrices}) that the
    scattering matrices depend on $\delta$ as in case of Dirac
    operator.
\end{remark}

Using this theorem, we show that the scattering problem for equation
(\ref{intro:eq:equation}) with arbitrary boundary condition
(\ref{intro:eq:boundary_condition}) can be reduced to the scattering problem
with the Dirichlet boundary condition. For any $\alpha\in[0,\pi)$ we introduce
the mappings $\xi_{\alpha}: {\mathcal{Q}} \to{\mathcal{Q}}$ and $\phi_{\alpha
}: {\mathcal{Q}} \to[0,\pi)$ by
\[
\begin{gathered}
\xi_{\a}(p,q) = (\tilde{p},\tilde{q}) \iff (\tilde{p},\tilde{q},0) \in \Iso(p,q,\a),\\
\phi_{\a}(p,q) = \a + \int_0^{+\iy}(p(x) - \tilde{p}(x)) dx \mod \pi.
\end{gathered}
\]

\begin{theorem}
\label{intro:thm:reducing_of_scattering_problem} Let $\alpha\in[0,\pi)$. Then
the mapping $\xi_{\alpha}: {\mathcal{Q}} \to{\mathcal{Q}}$ is a bijection and
for any $(p,q) \in{\mathcal{Q}}$ we have
\[
S_{\alpha}(\cdot,p,q) = e^{-2i\phi_{\alpha}(p,q)} S_{0}(\cdot,\xi_{\alpha
}(p,q)).
\]

\end{theorem}

\begin{remark}
    In Section \ref{proof}, the mappings $\xi_{\alpha}$ and $\phi_{\alpha}$ 
    will be constructed implicitly via
    (\ref{intro:eq:isoresonances_potentials}) and
    (\ref{intro:eq:isoresonances_potentials_ivp}). In Corollary
    \ref{dirac:cor:reducing_of_scattering_problem} we obtain similar result for
    Dirac operators. In contrast to Theorem
    \ref{intro:thm:reducing_of_scattering_problem}, the mappings $\xi_{\alpha}$ and 
    $\phi_{\alpha}$ in the case of the Dirac operator are
    significantly straightforward.
\end{remark}

\subsection{A short review}

There are a lot papers devoted to resonances for various operators see
articles \cite{F97, H99, K04a, S00, Z87} concerning to one-dimensional
Schr\"odinger operator with compactly supported potentials and the book
\cite{DZ19} and the references therein concerning to another operators. In
particular, the inverse resonances problem (uniqueness, reconstruction and
characterization) were solved by Korotyaev for a Schr\"odinger operator with a
compactly supported real-valued potential on the half-line \cite{K04a} and on
the real line \cite{K05}. See also Zworski \cite{Z01} and
Brown-Knowles-Weikard \cite{BKW03} concerning the uniqueness problem. The
stability problems for resonances of one-dimensional Schr{\"o}dinger operators
were considered in papers \cite{K04b, MSW10, B12}.

Resonances are discussed for various perturbations of one-dimensional
Schr\"odinger operator. For example, Sch\"odinger operator with periodic plus
a compactly supported potential were considered by Firsova \cite{F84},
Korotyaev \cite{K11} and Korotyaev-Schmidt \cite{KS12}. Resonances for
Schr\"odinger operator with a step-like potential were discussed by
Christiansen \cite{C06}. Stark operator with compactly supported potentials
was studied in \cite{FH21, K17, K18}. Resonances were also considered for
three and fourth order differential operator with compactly supported
coefficients on the line, see \cite{K19, BK19}. Note that the inverse
resonance scattering was discussed in \cite{K11, KS12, C06, K17}.

In our paper we use properties of resonances of one-dimensional Dirac
operator. For this operator, global estimates on resonances in terms of
potential were obtained by Korotyaev \cite{K14}. In the case of smooth
potentials, the asymptotic distribution of resonances and the forbidden domain
was obtained by Iantchenko and Korotyaev \cite{IK14a, IK14b}. They also
considered resonances of radial Dirac operator in \cite{IK15}. The inverse
resonances problem was solved by Korotyaev and Mokeev for Dirac operator with
compactly supported potentials on the half-line \cite{KM21} and on the
real-line \cite{KM23}. The stability problem for resonances of Dirac operator
on the half-line was solved by Mokeev \cite{M22}. In the case of
super-exponential decaying potentials the inverse resonances problem was
solved by Gubkin \cite{G22}. There is a number of papers dealing with other
related problems for the one-dimensional Dirac operators, for instance, the
resonances were studied in application to the Dirac fields in the black holes,
see, e.g., \cite{I17, I18}.

The properties of resonances of Schr\"odinger equations with linearly energy
dependent potentials have not yet been studied. However, there are a lot of
results about direct and inverse scattering problem for this equation. The
scattering theory of Schr\"odinger equations with linearly energy dependent
potentials was first discussed by Cornille \cite{Cornille70} and Weiss and
Scharf \cite{WS71}. The inverse scattering problem for Schr\"odinger equations
with linearly energy dependent potentials was solved by Jaulent and Jean
\cite{J72,JJ72,JJ76a,JJ76b} on the half-line and on the real line when the
potentials are real-valued and there are no bound states. Kamimura
\cite{K07a,K07b,K08a,K08b} improved the results of Jaulent and Jean to a wider
class of potentials, but when also there are no bound states. Further, Hryniv
and Manko \cite{HM20} solved the inverse scattering problem on the half-line
for class of Miura potentials, transforming Schr\"odinger equations with
linearly energy dependent potentials to Dirac operator on the half-line.

Using the Riemann-Hilbert problem approach the inverse scattering problem with
purely imaginary or real potential $p$ was considered by Sattinger and
Szmigielski \cite{SS95,SS96}. In general, they discussed the case when there
are no bound states, but they also considered the case when there is one bound
states. Moreover, they considered this equation in application to the
Camassa-Holm equation. The inverse scattering problem with purely imaginary
potential $p$ was also considered by Aktosun, Klaus and van der Mee
\cite{AM90,AM91,AKM98a,AKM98b}.

Our paper is organized as follows: In Section \ref{jost} we obtain estimates
of the Jost function for potentials from ${\mathcal{P}}$. In Section
\ref{comp} we prove the main theorems for potentials from ${\mathcal{P}}$. In
Section \ref{dirac} we give results about the inverse resonances problem for
the Dirac operator. In Section \ref{trans} we describe the transformation of
equation (\ref{intro:eq:equation}) to the Dirac equation. In Section
\ref{proof} we prove the main theorems for potentials from ${\mathcal{Q}}$.

\section{Jost solution} \label{jost}

In this section we study properties of the Jost solution $y_+(\cdot,k)$
of equation (\ref{intro:eq:equation}) satisfying (\ref{intro:eq:jostsol_def}) when $(p,q) \in \cP$.
Recall that for any $k \in \C$ the Jost solution $y_+(\cdot,k)$ is a solution of the following integral equation
$$
    y_+(x,k) = e^{ikx} +
    \int_x^{\g} \frac{\sin k(t - x)}{k} V(t,k) y_+(t,k) dt,\qq x \in \R_+,
$$
where $V(t,k) = q(t) + 2 k p(t)$ for all $(t,k) \in \R_+ \ts \C$.
In order to study properties of the Jost solution we introduce the
function $y(x,k) = e^{-ikx}y_+(x,k)$, which satisfies
\[
\label{jost:eq:integral_equation}
    y(x,k) = 1 + \int_x^{\g} G(t - x, k) V(t,k) y(t,k) dt,
    \qq (x,k) \in \R_+ \ts \C,
\]
where $G(t,k) = \frac{e^{2ikt} - 1}{2ik}$ for all $(t,k) \in \R_+
\ts \C$. We introduce the function $u(x) = e^{i \int_x^{\g} p(s)
ds}$ for all $0 \leq x \leq \g$, which is a solution of the initial
value problem
$$
    u'(x) = -i p(x) u(x),\qq x \geq 0,\qq u(\g) = 1,
$$
and then
\[ \label{jost:eq:1}
    u(x) = 1 + i \int_x^{\g} p(s) u(s) ds,\qq 0 \leq x \leq \g.
\]
We rewrite $y$ in the form $y(x,k) = u(x) + v(x,k)$. Then we have
\[ \label{jost:eq:2}
    u(x) + v(x,k) = 1 +  \int_x^{\g} G(t - x, k) V(t,k)
    u(t)dt + \int_x^{\g} G(t - x, k) V(t,k) v(t,k)dt
\]
Using $V(t,k) = q(t) + 2kp(t)$ and $2kG(t,k) = -ie^{2ikt} + i$, we get
\begin{multline*}
    \int_x^{\g} G(t - x, k) V(t,k) u(t)dt =
    \int_x^{\g} G(t - x, k) q(t) u(t)dt + \int_x^{\g} 2kG(t - x, k) p(t) u(t)dt\\
    = \int_x^{\g} G(t - x, k) q(t) u(t)dt - i \int_x^{\g} e^{2ik(t-x)}
    p(t) u(t)dt + i \int_x^{\g} p(t) u(t)dt.
\end{multline*}
Using also (\ref{jost:eq:1}), we get
\begin{multline} \label{jost:eq:3}
    \int_x^{\g} G(t - x, k) V(t,k) u(t)dt
    \\
    = \int_x^{\g} G(t - x, k) q(t) u(t)dt - i \int_x^{\g} e^{2ik(t-x)} p(t) u(t)dt + u(x) - 1.
\end{multline}
Substituting (\ref{jost:eq:3}) in (\ref{jost:eq:2}), we obtain
the following integral equation for the function $v$:
\[
\label{jost:eq:integral_equation_v}
    v(x,k) = v_o(x,k) + \int_x^{\g} G(t - x, k) V(t,k) v(t,k)dt,
\]
where
\[
\label{jost:eq:definition_v_o}
    v_o(x,k) = \int_x^{\g} G(t - x, k) q(t) u(t)dt - i \int_x^{\g} e^{2ik(t-x)} p(t) u(t)dt.
\]

At first, we study properties of the function $v_o$.
Recall that $\norm{f}_{1}$ and $\norm{f}_{w}$ was defined in (\ref{intro:eq:norms}) and
the constants $C_1$ and $C_2$ was defined in (\ref{intro:eq:constants}).
For any $(p,q) \in \cP$ and $k \neq 0$ we introduce the following functions:
\[ \label{jost:eq:omega_definition}
    \omega_o(k) = \min\left(\norm{q}_{w} + \norm{p}_{1}, \frac{C_1}{\abs{k}}\right),\qq
    \omega_1(k) = \min\left(\norm{q}_{w} + 2 \abs{k}\norm{p}_{w},
    \frac{\norm{q}_{1}}{\abs{k}} + 2\norm{p}_{1}\right),
\]
Below we need the following notation
$$
    k_- = \frac{1}{2}\left(\abs{\Im k} - \Im k\right)=
    \begin{cases}
        0 &\text{if $k \in \ol \C_+$}\\
        \abs*{\Im k}  &\text{if $k \in \C_-$}
    \end{cases}.
$$

\begin{lemma} \label{jost:lm:v_o_properties}
    Let $(p,q) \in \cP$ and let $v_o$ be given by
    (\ref{jost:eq:definition_v_o}). Then each $v_o(x,\cdot)$, $x \geq 0$, is
    an entire function of exponential type and has the following representation
    \[ \label{jost:eq:v_o_representation}
        v_o(x,k) = \int_0^{\g - x} e^{2iks}\left( - i p(s + x)u(s + x)
        + \int_{s + x}^{\g} q(t)u(t) dt\right) ds,
    \]
    for all $(x,k) \in [0,+\iy) \ts \C$, and the following estimate holds true:
    \[ \label{jost:eq:v_o_bounds}
        \abs{v_o(x,k)} \leq \omega_o(k) e^{2 (\g - x) k_-}.
    \]
\end{lemma}
\begin{proof}
    Substituting $G(t,k) = \int_0^{t} e^{2iks}ds$ in (\ref{jost:eq:definition_v_o}), we get
    $$
        v_o(x,k) = \int_x^{\g} \int_0^{t - x} e^{2iks}ds q(t) u(t)dt - i \int_x^{\g} e^{2ik(t-x)} p(t) u(t)dt.
    $$
    Then we obtain (\ref{jost:eq:v_o_representation}) by changing the order of integration in the last formula.
    Due to $\abs{u(x)} = 1$ for all $x \geq 0$, it follows from (\ref{jost:eq:v_o_representation}) that
    $$
        \abs{v_o(x,k)} \leq \int_0^{\g - x} \abs{e^{2iks}}
        \left(\int_{s + x}^{\g} \abs{q(t)} dt + \abs{p(s + x)}\right) ds.
    $$
    Since $\abs{e^{2iks}} \leq e^{2 (\g - x) k_-}$ for all $(s,k) \in [0,\g - x] \ts \C$, we have
    $$
        \abs{v_o(x,k)} \leq e^{2 (\g-x)k_-}
        \left(\int_0^{\g - x} \int_{s + x}^{\g} \abs{q(t)} dt ds + \int_0^{\g - x} \abs{p(s + x)}ds\right).
    $$
    Changing the order of integration, we get
    \[ \label{jost:eq:4}
        \abs{v_o(x,k)} \leq e^{2 (\g-x) k_-} (\norm{q}_{w} + \norm{p}_{1}),\qq (x,k) \in [0,+\iy) \ts \C.
    \]
    Using integration by parts in (\ref{jost:eq:v_o_representation}), we have
    \begin{multline*}
        v_o(x,k) = \frac{1}{2ik} \left(\int_x^{\g} q(t) u(t) dt - i p(x)u(x) + ip(\g) e^{2ik(\g - x)}\right.\\
        + \left.\int_0^{\g - x} e^{2iks}u(s+x)\left(q(s + x)+ i p'(s + x) + p^2(s + x)\right) \right) ds.
    \end{multline*}
    Recall that $\abs{u(s+x)} = 1$ and $\abs{e^{2iks}} \leq e^{2 (\g -
    x) k_-}$ for all $(s,k) \in [0,\g - x] \ts \C$. Then, we have
    $$
        \abs{v_o(x,k)} \leq
        \frac{1}{2\abs{k}} \left(\norm{q}_{1} + \abs{p(x)}
        + e^{2 (\g - x) k_-} \abs{p(\g)} + e^{2 (\g - x) k_-} \left(\norm{q}_{1}+ \norm{p'} +
        \norm{p^2}\right)\right) ds.
    $$
    Since $e^{2 (\g - x) k_-} \geq 1$ and $\abs{p(x)} \leq \abs{p(\g)} + \norm{p'}$, we get
    \[ \label{jost:eq:5}
        \abs{v_o(x,k)} \leq e^{2 (\g - x) k_-}
        \frac{\norm{q}_{1} + \abs{p(\g)} + \norm{p'} + \norm{p^2}}{\abs{k}} =
        e^{2 (\g - x) k_-} \frac{C_1}{\abs{k}}.
    \]
    Combining (\ref{jost:eq:4}) and (\ref{jost:eq:5}), we obtain (\ref{jost:eq:v_o_bounds}).
\end{proof}

A solution of integral equation (\ref{jost:eq:integral_equation_v})
can be constructed as follows
\[\label{jost:eq:v_iterations}
    \begin{aligned}
        v(x,k) &= \sum_{n \geq 0} v_n(x,k),\qq (x,k) \in \R_+ \ts \C,\\
        v_n(x,k) &= \int_x^{\g} G(t-x,k) V(t,k) v_{n-1}(t,k) dt,\qq n \geq 1.
    \end{aligned}
\]

\begin{lemma} \label{jost:lm:v_function}
    Let $(p,q) \in \cP$ and let $v$, $v_n$ be given by
    (\ref{jost:eq:v_iterations}). Then for any $x \in [0,+\iy)$ the
    functions $v(x,\cdot)$, $v_n(x,\cdot)$, $n \geq 0$, are entire of
    exponential type
    and for any $(x,k) \in [0,+\iy) \ts \C$ they satisfy the following estimates:
    \begin{align}
        \label{jost:eq:v_iterations_estimate}
        \abs{v_n(x,k)} &\leq \omega_o(k) e^{2 (\g - x) k_-}
        \frac{\omega_1(k)^n}{n!},\\
        \label{jost:eq:v_estimates}
        \abs{v(x,k)} &\leq \omega_o(k) e^{2 (\g - x) k_- + \omega_1(k)},\\
        \label{jost:eq:v_estimates2}
        \abs{v(x,k) - v_o(x,k)} &\leq \omega_o(k)\omega_1(k) e^{2 (\g - x) k_- + \omega_1(k)}.
    \end{align}
\end{lemma}
\begin{proof}
    It follows from (\ref{jost:eq:v_iterations}) that
    \[ \label{jost:eq:6}
        v_n(x,k)=\int_{D_n(x)} \left( \prod_{1 \leq j \leq n} G(t_j -
        t_{j-1}, k) V(t_j, k) \right) v_o(t_n,k)dt,\qq dt= dt_1 \ldots dt_n,
    \]
    where $D_n(x) = \set{(t_1,\ldots, t_n) \in \R_+^n \given x = t_0
    \leq t_1 \ldots \leq t_n \leq \g}$. Recall that $\abs{G(t,k)} \leq
    \frac{e^{2 t k_-}}{\abs{k}}$ for all $t \in [0,\g]$ and $k \in \C
    \setminus \{0\}$. Substituting this estimate in (\ref{jost:eq:6}) and
    using (\ref{jost:eq:v_o_bounds}), we get
    \[ \label{jost:eq:7}
        \begin{aligned}
            \abs{v_n(x,k)} &\leq \int_{D_n(x)} \left( \prod_{1 \leq j \leq n}
            \frac{e^{2 (t_j - t_{j-1}) k_-}}{\abs{k}} \abs{V(t_j, k)} \right)
            \omega_o(k) e^{2 (\g - t_n) k_-} dt\\
            &\leq \frac{\omega_o(k) e^{2 (\g - x) k_-}}{\abs{k}^n} \int_{D_n(x)}
            \prod_{1 \leq j \leq n} \abs{V(t_j, k)} dt \leq
            \omega_o(k) e^{2 (\g - x) k_-}\frac{\norm{V(\cdot,k)}^n}{\abs{k}^n n!}.
        \end{aligned}
    \]
    Now, we use another estimate $\abs{G(t,k)} \leq e^{2 t k_-}t$ for
    all $t \in [0,\g]$ and $k \in \C$. Substituting this estimate in
    (\ref{jost:eq:6}) and using (\ref{jost:eq:v_o_bounds}), we get
    \[ \label{jost:eq:8}
        \begin{aligned}
            \abs{v_n(x,k)} &\leq \int_{D_n(x)} \left( \prod_{1 \leq j \leq n}
            e^{2 (t_j - t_{j-1}) k_-}(t_j - t_{j-1}) \abs{V(t_j, k)} \right)
            \omega_o(k) e^{2 (\g - t_n) k_-} dt
            \\
            &\leq \omega_o(k) e^{2 (\g - x) k_-} \int_{D_n(x)} \prod_{1 \leq j \leq n}
            (t_j - t_{j-1}) \abs{V(t_j, k)} dt
            \\
            &\leq \omega_o(k) e^{2 (\g - x) k_-} \int_{D_n(x)}
            \prod_{1 \leq j \leq n} t_j \abs{V(t_j, k)} dt
            \leq \omega_o(k) e^{2 (\g - x) k_-} \frac{\norm{V(\cdot,k)}_{w}^n}{n!}.
        \end{aligned}
    \]
    Combining (\ref{jost:eq:7}) and (\ref{jost:eq:8}), we obtain (\ref{jost:eq:v_iterations_estimate}).
    Moreover, since $G(x,\cdot)$, $V(x,\cdot)$ and $v_o(x, \cdot)$ are
    entire functions, it follows from representation (\ref{jost:eq:6}) that
    the function $v_n(x,\cdot)$ is entire for any $n \geq 0$.  Due to
    estimate (\ref{jost:eq:v_iterations_estimate}) holds true, it follows that series in
    (\ref{jost:eq:v_iterations}) converges uniformly for $x \in [0,\g]$ on any bounded
    subset of $\C$ and then the function $v(x,\cdot)$ is entire.
    Using (\ref{jost:eq:v_iterations}) and (\ref{jost:eq:v_iterations_estimate}), we get
    $$
        \abs{v(x,k)} \leq \sum_{n \geq 0} \omega_o(k) e^{2 (\g - x) k_-}
        \frac{\omega_1(k)}{n!} = \omega_o(k) e^{2 (\g - x) k_- + \omega_1(k)},
    $$
    and
    $$
        \begin{aligned}
            \abs{v(x,k) - v_o(x,k)} &\leq \sum_{n \geq 1} \abs{v_n(x,k)}
            \leq \sum_{n \geq 1} \omega_o(k) e^{2 (\g - x) k_-} \frac{\omega_1(k)^n}{n!}
            \\
            &\leq \omega_o(k)\omega_1(k) e^{2 (\g - x) k_-}
            \sum_{n \geq 0} \frac{\omega_1(k)^n}{n! (n+1)}
            \leq \omega_o(k)\omega_1(k) e^{2 (\g - x) k_- + \omega_1(k)}.
        \end{aligned}
    $$
\end{proof}

Recall that we was looking for the Jost solution $y_+(\cdot,k)$ of equation (\ref{intro:eq:equation})
in the form $y_+(x,k) = e^{ikx} y(x,k)$, where $y(x,k) = u(x) + v(x,k)$ and $v$ is a solution of
integral equation (\ref{jost:eq:integral_equation_v}). So that Lemma \ref{jost:lm:v_function} implies
the following lemma.
\begin{lemma} \label{jost:lm:jost_solution}
    Let $(p,q) \in \cP$. Then for any $k \in \C$ there exists a unique Jost solution $y_+(\cdot,k)$
    of equation (\ref{intro:eq:equation}) satisfying (\ref{intro:eq:jostsol_def}).
    Moreover, the function $y_+(x,\cdot)$, $x \geq 0$, is an entire function of exponential type.
\end{lemma}

\section{Complex-valued potentials} \label{comp}

We estimate the imaginary part of an eigenvalue of our equation
(\ref{intro:eq:equation}).

\begin{theorem} \label{comp:thm:eigenvalue_bound}
    Let $(p,q) \in \cP$ and let $k_o \in \C_+$, $\abs{k_o} \geq 1$, be an eigenvalue of
    equation (\ref{intro:eq:equation}) with the Dirichlet boundary condition.
    Then it satisfies the following inequality:
    \[ \label{comp:eq:eigenvalue_and_resonance_bound}
        \Im k_o \leq \frac{\g}{2} (\norm{q}_2 + 2 \norm{p}_2)^4.
    \]
\end{theorem}

\begin{proof}
    We introduce the free Schr\"odinger operator $H_o y = -y''$
    on $L^2(\R_+)$ with the Dirichlet boundary condition $y(0)=0$.
    We define the Birman-Schwinger operator $Y(k)$, $k \in \C_+$, by
    $$
        Y(k) = V^{\frac{1}{2}} (H_o - k^2)^{-1} \abs{V}^{\frac{1}{2}},
    $$
    where $V^{\frac{1}{2}} =\abs{V}^{\frac{1}{2}} \sign V$.
    Recall that the resolvent $(H_o - k^2)^{-1}$ is an integral operator on the half-line with the kernel
    $$
        G(x,y,k) = \begin{cases}
            \frac{\sin(xk)}{k} e^{iky},\qq x \leq y\\
            \frac{\sin(yk)}{k} e^{ikx},\qq y \leq x
        \end{cases},\qq (x,y) \in \R_+^2.
    $$
    By simple calculation, we get
    $$
        G(x,y,k) = \frac{G_o(x,y,k)}{2ik},\qq
        G_o(x,y,k)= e^{ik(x+y)} - e^{ik \abs{x-y}},\qq
        (x,y) \in \R_+^2.
    $$
    Thus $Y(k)$ is an integral operator on the half-line with the kernel
    $$
        Y(x,y,k) = V^{\frac{1}{2}}(x,k)\frac{G_o(x,y,k)}{2ik}\abs{V(y,k)}^{\frac{1}{2}},\qq
        (x,y) \in \R_+^2.
    $$
    It is known that if $\norm{Y(k)} < 1$ for some $k \in \C_+$, then $k$
    is not an eigenvalue of equation (\ref{intro:eq:equation}) with the
    Dirichlet boundary condition. Using the H\"older's inequality and
    $\supp V(\cdot,k) \ss [0,\g]$, we get
    \[ \label{comp:eq:26}
        \begin{aligned}
            \norm{Y(k)}^2 &\leq \int_{\R_+^2} \abs{Y(x,y,k)}^2 dx dy =
            \frac{1}{4 \abs{k}^2}\int_{[0,\g]^2} \abs{V(x,k)} \abs{G_o (x,y,k)}^2 \abs{V(y,k)} dx dy\\
            &\leq \frac{1}{4\abs{k}^2} I(k)^\frac{1}{2} \int_{[0,\g]} \abs{V(x,k)}^2 dx =
            \frac{1}{4\abs{k}^2} I(k)^\frac{1}{2} \norm{V(\cdot,k)}_2^2,
        \end{aligned}
    \]
    where $I(k)=\int_{[0,\g]^2} \abs*{G_o (x,y,k)}^4 dx dy$.
   Now we estimate $I(k)$. Let $\Im k > 0$. Then we have
    $\abs{e^{ik \abs{x-y}}} \leq e^{-\Im k\abs{x-y}}$ for any $(x,y) \in [0,\g]^2$.
    Due to $x + y \geq \abs{x - y}$ we have
    $\abs{e^{ik (x+y)}} \leq e^{-\Im k\abs{x-y}}$ for any $(x,y) \in [0,\g]^2$.
    Using these estimates, we get
    \[ \label{comp:eq:24}
        \begin{aligned}
            I(k) = \int_{[0,\g]^2} \abs*{e^{ik(x+y)} - e^{ik \abs{x-y}}}^4 dx dy
            &\leq \int_{[0,\g]^2} (2e^{-\Im k\abs{x-y}})^4 dx dy\\
            &\leq 16 \int_{[0,\g]} dy \int_{\R}e^{- 4 \Im k \abs{x-y}} dx = \frac{8\g}{\Im k}.
        \end{aligned}
    \]
    Since $V(\cdot,k) = q + 2kp$, we get for $\abs{k} \geq 1$:
    \[ \label{comp:eq:25}
        \frac{\norm{V(\cdot,k)}_2}{\abs{k}} \leq \frac{\norm{q}_2}{\abs{k}} + 2 \norm{p}_2 \leq \norm{q}_2 + 2 \norm{p}_2.
    \]
    Substituting (\ref{comp:eq:24}) and (\ref{comp:eq:25}) in (\ref{comp:eq:26}), we get for $\abs{k} \geq 1$:
    $$
        \norm{Y(k)} \leq \frac{\g^{\frac{1}{4}}(\norm{q}_2 + 2 \norm{p}_2)}{(2\Im k)^{\frac{1}{4}}}.
    $$
    Thus, if we have
    $$
        \frac{\g^{\frac{1}{4}}(\norm{q}_2 + 2 \norm{p}_2)}{(2\Im k)^{\frac{1}{4}}} < 1,
    $$
    then $k$ is not an eigenvalue of equation (\ref{intro:eq:equation})
    with the Dirichlet boundary condition. This yields
    (\ref{comp:eq:eigenvalue_and_resonance_bound}) for an eigenvalue $k_o \in \C_+$, $\abs{k_o} \geq 1$.
\end{proof}

We are ready to prove the main theorems when $(p,q) \in \cP$.

\begin{proof}[Proof of Theorem \ref{intro:thm:eigenvalues_and_resonances_bounds}]
    Each eigenvalue $k_o \in \C_+ \cup \{0\}$ of equation (\ref{intro:eq:equation})
    with the Dirichlet boundary condition is a zero of the function $y_+(0,\cdot)$. Since 
    $y_+(0,\cdot) = y(0,\cdot)$, each eigenvalue $k_o$ is a solution of the equation
    \[ \label{comp:eq:7}
        y(0,k_o) = u(0) + v(0,k_o) = 0.
    \]
    Due to $\abs{u(0)} = 1$ if $\abs{v(0,k)} < 1$, then $k$ is not a solution of equation (\ref{comp:eq:7}).
    Using (\ref{jost:eq:v_estimates}) and (\ref{jost:eq:omega_definition}), we get
    $$
        \abs{v(0,k)} \leq \omega_o(k) e^{\omega_1(k)} \leq \frac{C_1}{\abs{k}} e^{\frac{\norm{q}_{1}}{\abs{k}} + 2\norm{p}_{1}} < 1.
    $$
    We rewrite the last inequality as follows
    \[ \label{comp:eq:3}
        \omega_* e^{\omega_*} < C_*,
    \]
    where
    \[ \label{comp:eq:4}
        \omega_* = \frac{\norm{q}_{1}}{\abs{k}},\qq C_* = \frac{\norm{q}_{1} e^{-2\norm{p}_{1}}}{C_1}.
    \]
    Let $W_o$ be the principal branch of the Lambert $W$ function,
    which satisfies the identity $W_o(x) e^{W_o(x)} = x$.
    Recall that $W_o$ is strictly increasing. So that, inequality (\ref{comp:eq:3}) implies
    \[ \label{comp:eq:8}
        \omega_* = \frac{\norm{q}_{1}}{\abs{k}} < W_o(C_*).
    \]
    Since $k_o$ is a solution of equation (\ref{comp:eq:7}) it follows from (\ref{comp:eq:8}) that
    $$
        \frac{\norm{q}_{1}}{\abs{k_o}} \geq W_o(C_*).
    $$
    Using (\ref{comp:eq:4}), we get
    \[ \label{comp:eq:5}
        \abs{k_o} \leq \frac{\norm{q}_{1}}{W_o(C_*)} = \frac{\norm{q}_{1}}{C_*} \frac{C_*}{W_o(C_*)} =
        C_1 e^{2\norm{p}_{1}} \frac{C_*}{W_o(C_*)}.
    \]
    It follows from definition of $W_o(x)$ that $\frac{x}{W_o(x)} = e^{W_o(x)}$.
    Substituting this identity in (\ref{comp:eq:5}), we obtain
    \[ \label{comp:eq:6}
        \abs{k_o} \leq C_1 e^{2\norm{p}_{1} + W_o(C_*)}.
    \]
    Note that $0 \leq C_* \leq e^{-2\norm{p}_{1}} \leq 1$. In particular, it leads to the estimate
    $$
        e^{W_o(C_*)} \leq e^{W_o(1)} = \frac{1}{W_o(1)} < 2.
    $$
    We need the following estimate (see Theorem 2.3 in \cite{HH08})
    $$
        W_o(x) \leq \log \frac{x + y}{1 + \log y},\qq x > -\frac{1}{e},\qq y > \frac{1}{e}.
    $$
    Putting $y = 1$, we get
    \[ \label{comp:eq:11}
        W_o(x) \leq \log (1 + x),\qq x > -\frac{1}{e}.
    \]
    Using this inequality, we get
    $$
        e^{W_o(C_*)} \leq e^{W_o(e^{-2\norm{p}_{1}})}
        \leq 1 + e^{-2\norm{p}_{1}}.
    $$
    Substituting this bound in (\ref{comp:eq:6}), we get
    $$
        \abs{k_o} \leq C_1 (1 + e^{2\norm{p}_{1}}).
    $$

    Let $r_o \in \C_-$ be a resonance of equation (\ref{intro:eq:equation})
    with the Dirichlet boundary condition. Then we have $y(0,r_o) = 0$.
    Since $y(x,k) = u(x) + v(x,k)$ it follows that
    $\abs{v(0,r_o)} = \abs{u(0)} = 1$. Using (\ref{jost:eq:v_estimates}), we get
    \[ \label{comp:eq:1}
        \omega_o(r_o) e^{2\g\abs{\Im r_o} + \omega_1(r_o)} \geq 1.
    \]
    It follows from (\ref{jost:eq:omega_definition}) that
    \[ \label{comp:eq:2}
        \omega_o(r_o) \leq \frac{C_1}{\abs{r_o}},\qq
        \omega_1(r_o) \leq C_2,
    \]
    where $C_1$ and $C_2$ were given in (\ref{intro:eq:constants}).
    Substituting (\ref{comp:eq:2}) in (\ref{comp:eq:1}), we get
    $$
        \abs{r_o} \leq C_1 e^{C_2 + 2\g\abs{\Im r_o}}.
    $$

    Estimate (\ref{intro:eq:imaginary_part_eigenvalue_bound}) was proved in Theorem \ref{comp:thm:eigenvalue_bound}.

    Let $(p,q) \in \cP$ be real-valued. Let $k_o$ be an eigenvalue of equation
    (\ref{intro:eq:equation}) with the Dirichlet boundary condition
    and let $e_o$ be the corresponding eigenfunction.
    Let $(\cdot,\cdot)$ be a scalar product in $L^2(\R_+,\C)$.
    Integrating by parts and using $e_o(0) = 0$ and
    $\displaystyle \lim_{x \to + \iy} e_o(x) = 0$, we get $(-e_o'', e_o) = \norm{e_o'}_2^2$.
    Then taking the imaginary part of the scalar product
    $(-e_o''+ V e_o, e_o)= k_o^2 \norm{e_o}_2^2$, we get
    $$
        \Re k_o = \frac{(p e_o, e_o)}{\norm{e_o}_2^2} =
        \frac{\int_{\R_+} p(x) \abs{e_o(x)}^2 dx}{\norm{e_o}_2^2}.
    $$
    Recall that $p_+$ and $p_-$ was defined by (\ref{intro:eq:constants}).
    Since $p \in C(\R_+)$, we get
    $$
        p_- \leq \Re k_o \leq p_+.
    $$

\end{proof}

\begin{proof}[Proof of Theorem \ref{intro:thm:resonances_estimate}]
 Recall that $y=u+v$.   Firstly, we find $k_o \in \C_+$ such that $v(0,k_o) \leq \frac{1}{2}$.
    Using (\ref{jost:eq:v_estimates}) and (\ref{jost:eq:omega_definition}), for any $k \in \C_+$ we get
    $$
        \abs{v(0,k)} \leq \omega_o(k) e^{\omega_1(k)} \leq \frac{C_1}{\abs{k}} e^{\frac{\norm{q}_{1}}{\abs{k}} + 2\norm{p}_{1}}.
    $$
    Thus, if we have
    \[ \label{comp:eq:9}
        \frac{C_1}{\abs{k}} e^{\frac{\norm{q}_{1}}{\abs{k}} + 2\norm{p}_{1}} \leq \a,
    \]
    then we get $v(0,k) \leq \a$. We rewrite (\ref{comp:eq:9}) as follows
    \[ \label{comp:eq:10}
        \omega_* e^{\omega_*} \leq C_*,\qq \omega_* = \frac{\norm{q}_{1}}{\abs{k}},\qq C_* = \a \frac{\norm{q}_{1} e^{-2\norm{p}_{1}}}{C_1}.
    \]
    Since the Lambert $W_o$ function is strictly increasing, inequality (\ref{comp:eq:10}) is equivalent
    to the inequality
    $$
        \omega_* \leq W_o(C_*).
    $$
    Using $\frac{\norm{q}_{1}}{C_*} = \frac{C_1}{\a} e^{2\norm{p}_{1}}$ and $\frac{x}{W_o(x)} = e^{W_o(x)}$, we rewrite the last inequality as follows
    $$
        \abs{k} \geq \frac{\norm{q}_{1}}{W_o(C_*)} = \frac{\norm{q}_{1}}{C_*} \frac{C_*}{W_o(C_*)} = \frac{C_1}{\a} e^{2\norm{p}_{1} + W_o(C_*)}.
    $$
    Note that $0 < C_* < \a e^{-2\norm{p}_{1}} < \a$. Using (\ref{comp:eq:11}), we get
    $$
        e^{W_o(C_*)} \leq 1 + \a,
    $$
    which yields
    $$
        \frac{C_1}{\a} e^{2\norm{p}_{1} + W_o(C_*)} \leq \frac{C_1(1 + \a)}{\a} e^{2\norm{p}_{1}}.
    $$
    Let $\a = \frac{1}{2}$. Then it follows from $\abs{k} \geq 3C_1e^{2\norm{p}_{1}}$
    that $\abs{v(0,k)} \leq \frac{1}{2}$.
    Let $k_o = i r_o$, where $r_o =3C_1e^{2\norm{p}_{1}}$, and let $F(k) = \psi(k+k_o) = y(0,k + k_o)$.
    Then $F$ is an entire function of exponential type and
     the Jensen formula (see e.g. \cite{Koosis98})
    holds true
    \[ \label{comp:eq:12}
        \log \abs{F(0)} + \int_0^{r} \frac{\cN(t,F)}{t} dt
        = \frac{1}{2\pi} \int_0^{2\pi} \log \abs{F(r e^{i\t})} d\t,
    \]
    where $F(0)\ne 0$.
    Since $\abs{v(0,k_o)} \leq \frac{1}{2}$ and $\abs{u(0)} = 1$, we have
    \[ \label{comp:eq:13}
        \log \abs{F(0)} = \log \abs{u(0) + v(0,k_o)}
        \geq \log \abs{\abs{u(0)} - \abs{v(0,k_o)}} \geq \log \frac{1}{2} = - \log 2.
    \]
    Since $\cN(\cdot,F)$ is an increasing function and $\cN(r - r_o, \psi) \leq \cN(r,F)$ for $r \geq r_o$, we get
    \[ \label{comp:eq:14}
        \int_0^{r} \frac{\cN(t,F)}{t} dt \geq \cN(r/2,F)
        \int_{r/2}^r \frac{dt}{t} \geq \cN(r/2 - r_o,\psi) \log 2.
    \]
    Now we estimate the right-hand side of the Jensen formulas
    (\ref{comp:eq:12}). Using (\ref{jost:eq:v_estimates}), we get
    \[ \label{comp:eq:15}
        \begin{aligned}
            \frac{1}{2\pi} \int_0^{2\pi} \log \abs{F(r e^{i\t})} d\t
            &\leq \abs{u(0)} + \frac{1}{2\pi} \int_0^{2\pi} \log \abs{v(0,
            k_o + r e^{i\t})} d\t
                        \\
            &\leq 1 + \frac{1}{2\pi} \int_{0}^{2\pi} \log \left( \omega_o(k_o + re^{i\t})
            e^{2\g(k_o + e^{i\t})_- + \omega_1(k_o + re^{i\t})} \right) d\t
            \\
            &\leq 1 + \frac{1}{2\pi} \int_{0}^{2\pi} \log \o_o(k_o + re^{i\t}) d\t\\
            &+ \frac{1}{2\pi} \int_{0}^{2\pi} \o_1(k_o + re^{i\t}) d\t + \frac{1}{2\pi}
            \int_{0}^{2\pi} 2\g(k_o + re^{i\t})_- d\t.
        \end{aligned}
    \]
    Using (\ref{jost:eq:omega_definition}), we get
    \[ \label{comp:eq:16}
        \begin{aligned}
            \frac{1}{2\pi} \int_{0}^{2\pi} \log \omega_o(k_o + re^{i\t}) d\t
            &\leq \log C_1 - \log(r - r_o),
            \\
            \frac{1}{2\pi} \int_{0}^{2\pi} \omega_1(k_o + re^{i\t}) d\t &\leq
            \frac{\norm{q}_{1}}{r - r_o} + 2\norm{p}_{1}.
        \end{aligned}
    \]
    Recall that $k_o = i r_o$ and $k_- = \abs{\Im k}$ if $k \in \C_-$ and $k_- = 0$ if $k \in \ol \C_+$. We have
    $$
        (k_o + re^{i\t})_- = (r \cos \t + i(r_o + r\sin \t))_-.
    $$
    Let $r \geq r_o$ and let $\t_o = \arcsin \frac{r_o}{r}$.
    Then $r_o + r\sin \t \leq 0$ if and only if $\pi + \t_o \leq \t \leq 2\pi - \t_o$, which yields
    \[ \label{comp:eq:17}
        \begin{aligned}
            \frac{1}{2\pi} \int_{0}^{2\pi} 2\g(k_o + re^{i\t})_- d\t
            &= -\frac{\g}{\pi} \int_{\pi + \t_o}^{2\pi-\t_o} (r_o + r \sin \t) d\t
            = \frac{\g r_o (2\t_o - \pi)}{\pi} + \frac{2 \g r \cos
            \t_o}{\pi}
            \\            &
            = \frac{2\g}{\pi} \sqrt{r^2 - r_o^2} + \frac{2
            \g}{\pi} r_o \arcsin \frac{r_o}{r} - \g r_o.
        \end{aligned}
    \]
    Substituting (\ref{comp:eq:16}) and (\ref{comp:eq:17}) in (\ref{comp:eq:15}), we get
    \[ \label{comp:eq:18}
        \begin{aligned}
            \frac{1}{2\pi} \int_0^{2\pi} \log \abs{F(r e^{i\t})} d\t
            &\leq \frac{2\g}{\pi} \sqrt{r^2 - r_o^2} - \log(r - r_o) + 2\norm{p}_{1} - \g r_o
            \\
            &+ 1 + \log C_1
            + \frac{\norm{q}_{1}}{r - r_o} + \frac{2 \g}{\pi} r_o
            \arcsin \frac{r_o}{r}.
        \end{aligned}
    \]
    Substituting (\ref{comp:eq:13}), (\ref{comp:eq:14}) and (\ref{comp:eq:18}) in (\ref{comp:eq:12}), we obtain
    $$
        \begin{aligned}
            \cN(r/2 - r_o,\psi) \log 2 - \log 2
            &\leq \frac{2\g}{\pi} \sqrt{r^2 - r_o^2} - \log(r - r_o) + 2\norm{p}_{1} - \g r_o
            \\
            &+ 1 + \log C_1
            + \frac{\norm{q}_{1}}{r - r_o} + \frac{2 \g}{\pi} r_o
            \arcsin \frac{r_o}{r}.
        \end{aligned}
    $$
    Let $R = \frac{r}{2} - r_o$. Then we have $r = 2R + 2r_o$ and
    \[ \label{comp:eq:19}
        \begin{aligned}
            \cN(R,\psi) \log 2 - \log 2
            &\leq \frac{2\g}{\pi} \sqrt{4 R^2 + 8Rr_o + 3r_o^2} - \log(2R + r_o) + 2\norm{p}_{1} - \g r_o
            \\
            &+ 1 + \log C_1
            + \frac{\norm{q}_{1}}{2R + r_o} + \frac{2 \g}{\pi} r_o
            \arcsin \frac{r_o}{2R + 2r_o}.
        \end{aligned}
    \]
    If $R \geq 0$, then we have
    \[ \label{comp:eq:20}
        \begin{aligned}
            \sqrt{4 R^2 + 8Rr_o + 3r_o^2} &\leq 2R + 2r_o,\qq
            \frac{\norm{q}_{1}}{2R + r_o} \leq \frac{\norm{q}_{1}}{r_o},\qq
            \arcsin \frac{r_o}{2R + 2r_o} \leq \frac{\pi}{6}\\
            - \log(2R + r_o) &\leq - \log r_o = - \log 3 - \log C_1 - 2\norm{p}_{1}.
        \end{aligned}
    \]
    Substituting (\ref{comp:eq:20}) in (\ref{comp:eq:19}), we get
    $$
        \begin{aligned}
            \cN(R,\psi) \log 2
            &\leq \frac{4\g}{\pi} R + \frac{4\g}{\pi} r_o - \log 3 - \log C_1 - 2\norm{p}_{1} + 2\norm{p}_{1} - \g r_o
            \\
            &+ 1 + \log C_1
            + \frac{\norm{q}_{1}}{r_o} + \frac{\g}{3} r_o
            + \log 2
            \\
            &= \frac{4\g}{\pi} R + \g r_o\left(\frac{4}{\pi} - \frac{2}{3} \right) + \frac{\norm{q}_{1}}{r_o} + 1 + \log 3 - \log 2.
        \end{aligned}
    $$
    Substituting $r_o = 3C_1e^{2\norm{p}_{1}}$ in this inequality and using $\frac{\norm{q}_{1}}{C_1} \leq 1$, we get
    $$
        \cN(R,\psi) \leq \frac{4 \g}{\pi \log 2} R + 
        \frac{3\left(\frac{4}{\pi} - \frac{2}{3}\right)}{\log 2} \g C_1 e^{2\norm{p}_{1}} + \frac{\frac{4}{3} + \log \frac{3}{2}}{\log 2}
        \leq \frac{4 \g}{\pi \log 2} R + C_3,
    $$
    where $C_3$ was given by (\ref{intro:eq:constant}).

    In order to estimate the number of eigenvalues in the case of real-valued potentials 
    we just recall that for an eigenvalue $\abs{k_o} > 1$ we have (\ref{intro:eq:imaginary_part_eigenvalue_bound})
    and (\ref{intro:eq:real_part_eigenvalue_bound}), which yields that $\abs{k_o} \leq R$, where
    $R$ is given by 
    $$
        R=\max\left(1, \max(p_+,p_-) + \frac{\g}{2} (\|q\|_2 + 2\|p\|_2)^4\right).
    $$
\end{proof}

\section{Inverse resonances problem for Dirac operator} \label{dirac}

We consider the self-adjoint Dirac operators $H_{\a}$,
$\a \in [0,\pi)$, on $L^2(\R_+,\C^2)$ given by
\[ \label{dirac:eq:differential_expression}
    H_{\a} \boldsymbol{w} = -i \s_3 \boldsymbol{w}' + i \s_3 Q \boldsymbol{w},\qq \boldsymbol{w} = \ma w_1 \\ w_2 \am,\qq \s_3 = \ma 1 & 0 \\ 0 & -1 \am,
\]
with the boundary condition
\[ \label{dirac:eq:bc}
    e^{-i\a} w_1(0) - e^{i\a} w_2(0) = 0.
\]
Here the matrix-valued potential $Q$ has the following form
$$
    Q = \ma 0 & v \\ \overline{v} & 0 \am,\qq v \in \cX.
$$
Recall that $\cX$ is the set of all functions $v \in L^2(\R_+,\C)$ such
that $\sup \supp v = \g$. 
Note that $H_{\a}$ is essentially self-adjoint on the subspace of functions from $C_o^{\iy}(\R_+,\C^2)$ 
satisfying (\ref{dirac:eq:bc}). Thus, it has a unique self-adjoint extension.
The spectrum of this operator is
absolutely continuous, i.e. $\s(H_{\a}) = \s_{ac}(H_{\a}) = \R$.

We introduce a vector-valued Jost solution $\by(x,k)$ of the Dirac equation
\[ \label{dirac:eq:dirac_equation}
    \by'(x,k) - Q(x)\by(x,k) = i\s_3 k \by(x,k),\qq (x,k) \in \R_+ \ts \C,
\]
satisfying condition
$$
    \by(x,k) = \ma y_1(x,k) \\ y_2(x,k) \am = \ma e^{i k x} \\ 0 \am,\qq x \geq \g.
$$
For each operator  $H_{\a}, \a \in [0,\pi)$, we introduce a Jost
function $\bpsi_{\a}:\C \to \C$ by
$$
    \bpsi_{\a}(k) = e^{-i\a} y_{1} (0,k) - e^{i\a} y_{2} (0,k),\qq k \in \C.
$$
It is well-known that $\bpsi_{\a}$ is entire, $\bpsi_{\a}(k) \neq 0$
for any $k \in \ol \C_+ $ and it has zeros in $\C_-$, which are
called \textit{resonances} of the operator $H_{\a}$. The
multiplicity of a resonance is the multiplicity of the zero of
$\bpsi_{\a}$. We also introduce a scattering matrix $\bS_{\a}: \R
\to \C$ by
$$
    \bS_{\a}(k) = \frac{\ol \bpsi_{\a}(k)}{\bpsi_{\a}(k)},\qq k \in \R.
$$
Note that $\bS_{\a}$ admit a unique meromorphic continuation from $\R$
onto $\C$ and then the resonances are poles of the scattering
matrix.

For any $\a \in [0,\pi)$, we introduce  the mapping $\br^{\a} : \cX
\mapsto \mathbf{s}$ such that $\br^{\a}(v) = (r_n)_{n \geq 1}$ is a
sequence of resonances of the operator $H_{\a}$ with potential $v
\in \cX$ counted with the multiplicity and arranged such that
$$
    \begin{cases}
        \abs{r_1} \leq \abs{r_2} \leq \ldots,\\
        \Re r_n \leq \Re r_{n+1}\qq \text{if $\abs{r_n} = \abs{r_{n+1}}$}.
    \end{cases}
$$
In order to formulate results about inverse problems for Dirac
operator,  we introduce classes of Jost functions, scattering
matrices and resonances. For the Jost functions of Dirac operator we
introduce the following class.
\begin{definition*}
    $\cJ_{\a}$, $\a \in [0,\pi)$, is a set of all entire functions $f$,
 which have no zeros in $\ol \C_+$ and admit the following
    representation
    \[ \label{dirac:eq:jost_representation}
        f(k) = e^{-i\a} + \int_{0}^{\g} h(s) e^{2iks} ds,\qq k \in \C,\qq h \in \cX.
    \]
\end{definition*}
Recall that the classes $\cR$, $\cS$ and $\cS_{\a}$, $\a \in [0,\pi)$, have been defined in Section \ref{intro}.
We need the following results (see Theorem 1.1 and Corollary 1.2 in \cite{KM21}).

\begin{theorem} \label{dirac:thm:inverse_jost}
    For any $\a \in [0,\pi)$ we have:
    \begin{enumerate}
        \item the mapping $v \mapsto \bpsi_{\a}(\cdot,v)$
        from $\cX$ into $\cJ_{\a}$ is a bijection between $\cX$ and $\cJ_{\a}$;
        \item the mapping $v \mapsto \bS_{\a}(\cdot,v)$
        from $\cX$ into $\cS_{\a}$ is a bijection between $\cX$ and $\cS_{\a}$.
    \end{enumerate}
\end{theorem}

\begin{theorem} \label{dirac:thm:inverse_resonances}
    For any $\a \in [0,\pi)$, the mapping $\br^{\a}: \cX \to \cR$ is a bijection.
\end{theorem}

Now we recall some properties of entire functions which are useful
to study Jost functions. We introduce a Cartwright class of entire
functions $\cE_{Cart}$.

\begin{definition*}
    $\cE_{Cart}$ is the set of all entire functions $f$ of exponential type such that
    $$
        \int_{\R} \frac{\log(1+|f(x)|)dx}{1 + x^2} < \iy,\qq \t_+(f) = 0,\qq \t_-(f) = 2\g>0,
    $$
where the type $\t_{\pm}(f) = \lim \sup_{y \to +\iy} \frac{\log
|g(\pm i y)|}{y}$.
\end{definition*}

Firstly, a function from $\cE_{Cart}$ admits  the Hadamard
factorization (see, e.g., p.130 in \cite{L96}).

\begin{theorem} \label{dirac:thm:Hadamard}
    Let $f \in \cE_{Cart}$ and let $f(0) \neq 0$.
    Let $(z_n)_{n \geq 1}$ be the zeros of $f$ in $\C \sm \{0\}$ counted with multiplicity and
    arranged such that $0 < |z_1| \leq |z_2| \leq \ldots$.
    Then $f$ has the Hadamard factorization
    \[ \label{dirac:eq:Hadamard_factorization}
        f(k) = f(0) e^{i \g k} \lim_{r \to +\iy}
            \prod_{|z_n| \leq r} \left(1 - \frac{k}{z_n}\right),\qq k \in \C,
    \]
    where the infinite product converges uniformly on compact subsets of $\C$ and
    $$
        \begin{aligned}
            \sum_{n \geq 1} \frac{|\Im z_n|}{|z_n|^2} &< +\iy,\\
            \exists \lim_{r \to +\iy} \sum_{\abs{z_n} \leq r} \frac{1}{z_n} &\neq \iy.
        \end{aligned}
    $$
\end{theorem}

Secondly, we need the Levinson result about zeros of functions  from
$\cE_{Cart}$ (see, e.g., p. 58 in \cite{Koosis98}). Recall that
$\cN(r,f)$ is the counting function of zeros of a function $f$ counted with
multiplicities in the disc of the radius $r$ centered at the origin.

\begin{theorem} \label{dirac:thm:Levinson}
    Let $f \in \cE_{Cart}$. Then we have
    \[ \label{dirac:eq:Levinson_asymptotics}
        \cN(r,f) = \frac{2\g}{\pi} r + o(r)\qq \text{as $r \to +\iy$}.
    \]
\end{theorem}

Since any Jost function admit representation (\ref{dirac:eq:jost_representation}) it follows that
$\bpsi_{\a} \in \cE_{Cart}$ and we get the following Corollary.

\begin{corollary} \label{dirac:cor:cartwright}
    Let $(v,\a) \in \cX \ts [0,\pi)$. Then $\bpsi_{\a}(\cdot,v) \in \cE_{Cart}$ and
    it satisfies (\ref{dirac:eq:Hadamard_factorization})-(\ref{dirac:eq:Levinson_asymptotics}).
\end{corollary}

It follows from Theorem \ref{dirac:thm:inverse_resonances}  that for
some sequence of resonances there exist a unique potential for any
boundary condition parameter. Moreover, for any $(v_o,\a_o) \in \cX
\ts [0,\pi)$ we introduce the set of isoresonances potentials and
boundary condition parameters $\Iso(v_o, \a_o)$ as follows
$$
    \Iso(v_o, \a_o) = \set*{(v,\a) \in \cX \ts [0,\pi) \given \br^{\a}(v) = \br^{\a_o}(v_o)}.
$$
In case of Dirac operator, this set is simple and we describe it in the following theorem.

\begin{theorem} \label{dirac:thm:isoresonances_potentials}
    Let $(v_o,\a_o) \in \cX \ts [0,\pi)$. Then we have
    $$
        \Iso(v_o,\a_o) = \set*{(v_{\d}, \a_{\d}) \given \d \in [0,\pi)},
    $$
    where
    $$
        v_{\d} = e^{2i\d}v_o,\qq \a_{\d} = \a_o + \d \mod \pi,\qq \d \in [0,\pi).
    $$
    Moreover, for any $\d \in [0,\pi)$, the following identities hold true:
    $$
        \bpsi_{\a_o + \d}(\cdot,e^{2i\d}v_o) = e^{-i\d}
        \bpsi_{\a_o}(\cdot,v_o),\qq \bS_{\a_o + \d}(\cdot,e^{2i\d}v_o) =
        e^{2i\d} \bS_{\a_o}(\cdot,v_o).
    $$
\end{theorem}
\begin{proof}
    Let $\z = \br^{\a_o}(v_o)$ for some $(v_o,\a_o) \in \cX \ts [0,\pi)$.
    Then $\z$ is a sequence of zeros of the corresponding Jost function
    $\bpsi_o = \bpsi_{\a_o}(\cdot,v_o) \in \cJ_{\a_o}$. Let $\d \in
    [0,\pi)$ and $\bpsi_{\d} = e^{-i\d} \bpsi_o$. Then $\z$ is a sequence
    of zeros of $\bpsi_{\d}$. Since $\bpsi_o \in \cJ_{\a_o}$, it follows
    that $\bpsi_{\d}$ has no zeros in $\ol \C_+$ and it has the following
    representation
    $$
        \bpsi_{\d}(k) = e^{-i\d} \left( e^{-i\a_o}
        + \int_0^{\g} h_o(s) e^{2iks} ds\right) = e^{-i(\a_o + \d)} +
        \int_0^{\g} h_{\d}(s) e^{2iks} ds,
    $$
    where $h_{\d} = e^{-i\d} h_o \in \cX$. Thus, we have $\bpsi_{\d} \in
    \cJ_{\a_o + \d}$ and, by Theorem \ref{dirac:thm:inverse_jost}, there
    exists a unique $v_{\d} \in \cX$ such that $\bpsi_{\a_o +
    \d}(\cdot,v_{\d}) = \bpsi_{\d} = e^{-i\d} \bpsi_o$ and
    $\br^{\a_o+\d}(v_{\d}) = \z$, which yields $(v_{\d},\a_o + \d) \in
    \Iso(v_o,\a_o)$. Moreover, due to Theorem
    \ref{dirac:thm:inverse_resonances} for any $\d \in [0,\pi)$ there
    exists a unique potential $v_{\d}$ such that
    $\br^{\a_o + \d}(v_{\d}) = \z$. So that we have
    $$
        \Iso(v_o,\a_o) = \set*{(v_{\d},\a_o + \d) \given \d \in [0,\pi)},
    $$
    where $v_{\d}$ was defined above.

    In order to recover $v_{\d}$ for any $\d \in [0,\pi)$, we construct
    a solution of the corresponding     Gelfand-Levitan-Marchenko
    equation. Let $\bS_o = \bS_{\a_o}(\cdot,v_o)$ be the scattering matrix
    of $H_{\a_o}(v_o)$ and $\bS_{\d} = \bS_{\a_o + \d}(\cdot,v_{\d})$ be the
    scattering
    matrix of $H_{\a_o + \d}(v_{\d})$.
    Since $\bS_{o} \in \cS_{\a_o}$ and $\bS_{\d} in \cS_{\a_o + \d}$, it follows from the definition of $\cS_{\a}$ that
    $$
        \begin{aligned}
            \bS_o(k) = e^{2i\a_o} + \int_{-\g}^{+\iy} F_o(s) e^{2iks} ds,\qq
            \bS_{\d}(k) = e^{2i(\a_o + \d)} + \int_{-\g}^{+\iy} F_{\d}(s) e^{2iks} ds
        \end{aligned}
    $$
    for some $F_o, F_{\d} \in L^2(\R_+) \cap L^1(\R_+)$.
    Since we have
    $$
        \bS_{\d} = \frac{\ol{\bpsi_{\d}}}{\bpsi_{\d}} = e^{2i\d}
        \frac{\ol{\bpsi_{o}}}{\bpsi_{o}} = e^{2i\d} \bS_o
    $$
    it follows that $F_{\d} = e^{2i\d} F_o$.
    We also introduce the following matrix-valued functions
    $$
        \Omega_{o}(x) = \ma 0 & F_o(-x)
        \\
        \ol{F_o(-x)} & 0\am,\qq \Omega_{\d}(x) = \ma 0 & F_{\d}(-x)
        \\ \ol{F_{\d}(-x)} & 0\am.
    $$
    Note that
    \[ \label{dirac:eq:eq1}
        \Omega_{\d} = e^{i\d \s_3} \Omega_o e^{-i\d \s_3}.
    \]

    We also introduce a matrix-valued Jost solution $\boldsymbol{f}(x,k)$ of the Dirac equation (\ref{dirac:eq:dirac_equation}) satisfying
    $$
        \boldsymbol{f}(x,k) = e^{ixk\s_3},\qq x \geq \g.
    $$
    Note that $\boldsymbol{f}$ can be constructed via the vector-valued Jost solution$\by$ as follows
    $$
        \boldsymbol{f}(x,k) = \left( \by(x,k) \mid \s_1 \ol{\by(x,\ol k)}\right),\qq \s_1 = \ma 0 & 1 \\ 1 & 0\am.
    $$
    We need the following results (see, e.g., Lemmas 2.1 and 2.5 in \cite{KM21}).
    The matrix-valued Jost solution $\boldsymbol{f}(x,k)$ of the Dirac equation with the potential $v_o \in \cX$ has the following representation
    $$
        \boldsymbol{f}(x,k) = e^{ixk\s_3} + \int_0^{\g} \G^o(x,k) e^{i(2s + x)k \s_3}
        ds,\qq \G^o = \ma \G^o_{11} & \G^o_{12} \\ \G^o_{21} & \G^o_{22}\am,
    $$
    where $v_o(x) = - \G^o_{12}(x,0)$ for almost all $x \in \R_+$.
    The function $\G^o$ can be found as a unique solution of the Gelfand-Levitan-Marchenko equation:
    \[ \label{dirac:eq:eq2}
        \G^o(x,s) + \Omega_o(x+s) + \int_0^{+\iy} \G^o(x,t) \Omega_o(x+t+s) dt = 0,
    \]
    where the identity holds true for almost all $x,s \in \R_+$.

    We introduce $\G^{\d} = e^{i\d \s_3} \G^o e^{-i\d \s_3}$. Using
    (\ref{dirac:eq:eq1}) and (\ref{dirac:eq:eq2}), we get
    $$
        \G^{\d}(x,s) + \Omega_{\d}(x+s) + \int_0^{+\iy} \G^{\d}(x,t) \Omega_{\d}(x+t+s) dt = 0,
    $$
    for almost all $x,s \in \R_+$.
    Then we recover $v_{\d}$ as follows:
    $$
        v_{\d}(x) = - \G^{\d}_{12}(x,0) = - e^{2i\d} \G^o_{12}(x,0) = e^{2i\d} v_o(x)
    $$
    for almost all $x \in \R_+$.
\end{proof}

As an immediate corollary of the last theorem we get that the scattering problem for Dirac operator with arbitrary
boundary condition (\ref{dirac:eq:bc}) can be reduced to the scattering problem with the Dirichlet boundary
condition.

\begin{corollary} \label{dirac:cor:reducing_of_scattering_problem}
    Let $\a \in [0,\pi)$. Then for any $v \in \cX$ we have
    $$
        \bS_{\a}(\cdot,v) = e^{2i\a} \bS_{0}(\cdot,e^{-2i\a}v).
    $$
\end{corollary}

Recall that $\cR_o$ is a class of sequences of resonances, which are symmetric with respect
to the imaginary line. We show that in case of the Neumann boundary condition, i.e.
when $\a = \frac{\pi}{2}$ in (\ref{dirac:eq:bc}), the resonances are symmetric with respect
to the imaginary line if and only if the potential is real-valued. Let $\cX_{real}$
be a subset of all real-valued functions from $\cX$.
We also need the following result (see Theorem 1.6 in \cite{KM21}).

\begin{theorem} \label{dirac:thm:conjugate_potential}
    Let $(v,\a) \in \cX \ts [0,\pi)$. Then we have
    $$
        \ol{\bpsi_{\a}(-\ol{k}, v)} = e^{2i\a} \bpsi_{\a}(k, e^{4i\a}\ol{v}),\qq k \in \C.
    $$
\end{theorem}

\begin{remark}
    In particular, for $\a = \frac{\pi}{2}$, we get
    \[ \label{dirac:eq:jost_function_perturbation}
        \ol{\bpsi_{\frac{\pi}{2}}(-\ol{k}, v)} = - \bpsi_{\frac{\pi}{2}}(k, \ol{v}),\qq k \in \C.
    \]
\end{remark}

\begin{theorem} \label{dirac:thm:inverse_resonances_real_valued_potentials}
    The mapping $\br^{\frac{\pi}{2}}: \cX_{real} \to \cR_o$ is a bijection.
\end{theorem}
\begin{proof}
    Using Theorem \ref{dirac:thm:inverse_resonances} and $\cX_{real} \ss \cX$, 
    we only need to prove that the mapping
    $\br^{\frac{\pi}{2}}: \cX_{real} \to \cR_o$ is a surjection. Let $v \in \cX_{real}$.
    Due to (\ref{dirac:eq:jost_function_perturbation}) we have
    $$
        \bpsi_{\frac{\pi}{2}}(k, v) = -\ol{\bpsi_{\frac{\pi}{2}}(-\ol{k}, v)},
    $$
    which yields $\br^{\frac{\pi}{2}}(v) \in \cR_o$.

    Let $(r_n)_{n \geq 1} \in \cR_o$. By Theorem \ref{dirac:thm:inverse_resonances} there exists
    a unique $v \in \cX$ such that $\br^{\frac{\pi}{2}}(v) = (r_n)_{n \geq 1}$.
    Due to Corollary \ref{dirac:cor:cartwright}, it
    follows from (\ref{dirac:eq:Hadamard_factorization}) that
    \[ \label{dirac:eq:2}
        \bpsi_{\frac{\pi}{2}}(k,v) = \bpsi_{\frac{\pi}{2}}(0,v) e^{i\g k}
        \lim_{r \to +\iy} \prod_{\abs{r_n} \leq r} \left(1 - \frac{k}{r_n}\right),
    \]
    which yields
    \[ \label{dirac:eq:1}
        -\ol{\bpsi_{\frac{\pi}{2}}(-\ol{k},v)} = -\ol{\bpsi_{\frac{\pi}{2}}(0,v)} e^{i\g k}
        \lim_{r \to +\iy} \prod_{\abs{r_n} \leq r} \left(1 - \frac{k}{-\ol{r_n}}\right).
    \]
    Since $(r_n)_{n \geq 1}$ is symmetric with respect to the imaginary line, we get
    $$
        \lim_{r \to +\iy} \prod_{\abs{r_n} \leq r} \left(1 - \frac{k}{-\ol{r_n}}\right) =
        \lim_{r \to +\iy} \prod_{\abs{r_n} \leq r} \left(1 - \frac{k}{r_n}\right).
    $$
    Substituting this identity in (\ref{dirac:eq:1}) and using (\ref{dirac:eq:2}), we obtain
    \[ \label{dirac:eq:3}
        -\ol{\bpsi_{\frac{\pi}{2}}(-\ol{k},v)} = -\frac{\ol{\bpsi_{\frac{\pi}{2}}(0,v)}}{\bpsi_{\frac{\pi}{2}}(0,v)} \bpsi_{\frac{\pi}{2}}(k,v).
    \]
    Now we describe the asymptotic of $\bpsi_{\frac{\pi}{2}}(k,v)$ at infinity.
    Since $\bpsi_{\frac{\pi}{2}}(\cdot,v) \in \cJ_{\frac{\pi}{2}}$ it follows from (\ref{dirac:eq:jost_representation}) that
    $$
        \bpsi_{\frac{\pi}{2}}(k,v) = -i + \int_0^{\g} h(s) e^{2iks} ds,\qq k \in \C, h \in \cX.
    $$
    Applying the Riemann-Lebesgue lemma (see e.g. Theorem IX.7 in \cite{RS80}), we get
    $$
        \lim_{k \to \pm \iy} \bpsi_{\frac{\pi}{2}}(k,v) = -i,
    $$
    which yields
    $$
        \lim_{k \to \pm \iy} -\ol{\bpsi_{\frac{\pi}{2}}(-\ol{k},v)} = -i.
    $$
    Going to the limit at $k \to \pm \iy$ in (\ref{dirac:eq:3}), we obtain $-\ol{\bpsi_{\frac{\pi}{2}}(0,v)} = \bpsi_{\frac{\pi}{2}}(0,v)$.
    Substituting this identity in (\ref{dirac:eq:3}), we get $-\ol{\bpsi_{\frac{\pi}{2}}(-\ol{k},v)} = \bpsi_{\frac{\pi}{2}}(k,v)$.
    On the other hand, using (\ref{dirac:eq:jost_function_perturbation}), we get
    $$
        \bpsi_{\frac{\pi}{2}}(k,v) = -\ol{\bpsi_{\frac{\pi}{2}}(-\ol{k},v)} = \bpsi_{\frac{\pi}{2}}(k,\ol{v}),\qq k \in \C.
    $$
    Due to Theorem \ref{dirac:thm:inverse_jost} the mapping $v \mapsto \bpsi_{\frac{\pi}{2}}(\cdot,v)$
    is a bijection, which yields $v = \ol{v}$ and then $v \in \cX_{real}$.
\end{proof}

At last we need the following result (see Theorem 1.3 in \cite{KM21}), which
describes the position of resonances and the forbidden domain.

\begin{theorem} \label{dirac:thm:forbidden_domain}
Let $(z_n)_{n \geq 1} = \br^{\a}(v)$ for some $(v,\a) \in \cX \ts
[0,\pi)$ and let $\ve > 0$.  Then there exists a constant $C =
C(\ve,v,\a) \geq 0$ such that the following inequality holds true
for each $n \geq 1$:
    $$
        2 \g \Im z_n \leq \ln \left( \ve + \frac{C}{\abs{z_n}} \right).
    $$
    In particular, for any $A > 0$, there are only finitely many resonances in the strip
    $$
        \set*{z \in \C \given 0 > \Im z > -A}.
    $$
\end{theorem}

\section{Transformation to Dirac equation} \label{trans}

In this section we recall a correspondence between energy dependent
Schr\"odinger equations and Dirac equations on the half-line. It was
shown in \cite{HM20} that there exists a bijection between
potentials of energy dependent Schr\"odinger equations and Dirac
equations such that the corresponding solutions of these equations
are connected in simple way.

Recall that $y$ is a solution of the energy dependent Schr\"odinger equation
\[ \label{trans:eq:Schrodinger_equation}
    -y'' + q y + 2pky = k^2 y
\]
for some $(p,q) \in \cQ$, where $q = u' + u^2$, $u \in L^2(0,\g)$, if
$y(\cdot,k)$ and its quasi-derivative $y^{[1]}(\cdot,k) =
y'(\cdot,k) - u(\cdot)y(\cdot,k)$ are locally absolutely continuous
functions and the following identity holds true:
$$
    -\left(y^{[1]}\right)' - uy^{[1]} + 2kpy = k^2y.
$$
Recall that Dirac equation has the following form
\[ \label{trans:eq:Dirac_equation}
    \boldsymbol{w}' + Q\boldsymbol{w} = ik\s_3 \boldsymbol{w},\qq Q = \ma 0 & v \\ \ol v & 0\am,\qq \boldsymbol{w} = \ma w_1 \\ w_2 \am.
\]
Recall that the mapping $\cT : \cQ \to \cX$ was defined in (\ref{intro:eq:potential_mapping}) as
$$
    \cT: (p,q) \mapsto (-u + ip) e ^{-2i\vp},\qq q = u' + u^2,\qq \vp(x)
    = \int_x^{+\iy} p(t) dt,\qq x \in \R_+.
$$
We need the following results from \cite{HM20}.
\begin{lemma} \label{trans:lm:solutions}
    Let $v = \cT(p,q)$ for some $(p,q) \in \cQ$ and let $\vp(x) = \int_x^{+\iy} p(t) dt$.
    Then the following statements hold true:
    \begin{enumerate}
        \item If $y$ is a solution of (\ref{trans:eq:Schrodinger_equation}), then the function $y$ given by
        $$
            \boldsymbol{w}(x,k) = \ma w_1(x,k) \\ w_2(x,k) \am =
            \ma -ie^{-i\vp(x)}  & e^{-i\vp(x)} \\ ie^{i\vp(x)} & e^{i\vp(x)}\am
            \ma k^{-1} y^{[1]}(x,k) \\ y(x,k)\am
        $$
        is a solution of (\ref{trans:eq:Dirac_equation}) for any $(x,k) \in \R_+ \ts \C \sm \{0\}$.

        \item If $\boldsymbol{w} = \ma w_{1} \\ w_{2}\am$ is a solution of (\ref{trans:eq:Dirac_equation}), then
        \[ \label{trans:eq:Schrodinger_solution}
            \ma k^{-1} y^{[1]}(x,k) \\ y(x,k)\am =
            \ma ie^{i\vp(x)}  & -ie^{-i\vp(x)} \\ e^{i\vp(x)} & e^{-i\vp(x)}\am
            \ma w_1(x,k) \\ w_2(x,k) \am
        \]
        is a solution of (\ref{trans:eq:Schrodinger_equation}) for any
        $(x,k) \in \R_+ \ts \C \sm \{0\}$.

        \item Let $y$ and $\boldsymbol{w} = \ma w_{1} \\ w_{2}\am$ be solutions of (\ref{trans:eq:Schrodinger_equation})
        and (\ref{trans:eq:Dirac_equation})
        and let identity (\ref{trans:eq:Schrodinger_solution}) hold true.
        Then we have
        $$
            \begin{aligned}
                e^{-i\b}w_1(0,k) - e^{i\b}w_2(0,k) &= 0,\\
                y^{[1]}(0,k) \sin \a + k y(0,k) \cos \a &= 0
            \end{aligned}
        $$
        for some $\a, \b \in [0,\pi)$ and $k \in \C \sm \{0\}$ if and only if
        $$
            \vp(0) + \a + \b = \frac{\pi}{2} \mod \pi.
        $$
    \end{enumerate}
\end{lemma}

Now, we show that the mapping $\cT$ is a bijection.

\begin{lemma} \label{trans:lm:potential_mapping}
    The mapping $\cT: \cQ \to \cX$ given by (\ref{intro:eq:potential_mapping}) is a bijection.
\end{lemma}
\begin{proof}
    We introduce the following sets:
    $$
        \begin{aligned}
            \cX_+ &= L^2(\R_+,\C) \cap L^1(\R_+,\C),\qq \cX_{+,real}
            = L^2(\R_+) \cap L^1(\R_+)\\
            \cQ_{+} &= \set*{(p,q) \in \cX_{+,real} \ts W^{-1,2}(\R_+) \given
            q = u' + u^2,\, u \in \cX_{+,real}}.
        \end{aligned}
    $$
    It was proved in Theorem 1 \cite{HM20}
    that the mapping $\cT_+:\cQ_{+} \to \cX_{+}$ given by (\ref{intro:eq:potential_mapping})
    is a bijection. Note that the following identities hold true:
    \[ \label{trans:eq:6}
        \begin{aligned}
            \cX &= \set*{v \in \cX_+ \given \sup \supp v = \g}\\
            \cQ &= \set*{(p,q) \in \cQ_+ \given \sup \supp (\abs{p} + \abs{u}) =
            \g,\, q = u' + u^2,\, u \in \cX_{+,real}},
        \end{aligned}
    \]
    and we have $\subs*{\cT_+}_{\cQ} = \cT$.
    So that in order to prove that the mapping $\cT$ is a bijection, we only need to prove that
    it is a surjection.

    Let $v = \cT(p,q)$ for some $(p,q) \in \cQ$, where $q = u' + u^2$, $u \in L^2(0,\g)$,
    and let $\vp(x) = \int_x^{+\iy} p(t) dt$.
    Using (\ref{intro:eq:potential_mapping}), we get
    \[ \label{trans:eq:4}
        \begin{cases}
            \Im v = p \cos(2\vp) + u \sin(2\vp),\\
            \Re v = -u \cos(2\vp) + p \sin(2\vp),
        \end{cases}
    \]
    which yields
    \[ \label{trans:eq:3}
        \sup \supp v \leq \sup \supp (\abs{p}+\abs{u}).
    \]
    Solving linear system (\ref{trans:eq:4}), we get
    $$
        \begin{cases}
            p = \Im v \cos(2\vp) + \Re v \sin(2\vp),\\
            u = -\Re v \cos(2\vp) + \Im v \sin(2\vp),
        \end{cases}
    $$
    which yields
    \[ \label{trans:eq:5}
        \sup \supp (\abs{p}+\abs{u}) \leq \sup \supp v.
    \]
    Combining (\ref{trans:eq:3}) and (\ref{trans:eq:5}), we get
    \[ \label{trans:eq:7}
        \sup \supp v = \sup \supp (\abs{p}+\abs{u}).
    \]
    By (\ref{trans:eq:6}), it follows from (\ref{trans:eq:7})
    that $v \in \cX$ if and only if $(p,q) \in \cQ$.
    Thus, the mapping $\cT: \cQ \to \cX$ is a bijection.
\end{proof}

\begin{remark}
    Note that for given $v \in \cX$ we can recover $(p,q) = \cT^{-1}(v)$ by the formulas:
    \[ \label{trans:eq:potentials_transform}
        \begin{aligned}
            q &= u' + u^2,\\
            u &= -\Re v \cos (2\vp) + \Im v \sin (2\vp),\\
            p &= \Im v \cos (2\vp) + \Re v \sin (2\vp),
        \end{aligned}
    \]
    where $\vp$ is a unique solution of the initial value problem
    \[ \label{trans:eq:potential_phase}
        \begin{cases}
            \vp' = -\Im v \cos (2\vp) - \Re v \sin (2\vp),\\
            \vp(\g) = 0.
        \end{cases}
    \]
\end{remark}

Recall that for any $k \in \C$ the Jost solution $y_{+}(\cdot,k)$
of equation (\ref{trans:eq:Schrodinger_equation}) satisfies
$$
    y_{+}(x,k) = e^{i k x},\qq x \geq \g.
$$
\begin{lemma} \label{trans:lm:jost_solution}
    Let $(p,q) \in \cQ$ and let $k \in \C \sm \{0\}$. Then there exists
    a Jost solution  $y_{+}(\cdot,k)$ of equation
    (\ref{trans:eq:Schrodinger_equation}) given by
    (\ref{trans:eq:Schrodinger_solution}),  where $\by$ is the
    vector-valued Jost solution of Dirac equation
    (\ref{trans:eq:Dirac_equation}) with $v = \cT(p,q)$.
\end{lemma}
\begin{proof}
    Let $(p,q) \in \cQ$ and let $k \in \C \sm \{0\}$. Let $\by = \ma
    y_{1}\\ y_{2}\am$ be the Jost solution of Dirac equation
    (\ref{trans:eq:Dirac_equation}) with $v = \cT(p,q)$ and let
    $y_{+}(\cdot,k)$ be given by (\ref{trans:eq:Schrodinger_solution}).
    Due to Lemma \ref{trans:lm:solutions} the function $y_{+}(\cdot,k)$
    is a solution of equation (\ref{trans:eq:Schrodinger_equation}) and
    it follows from (\ref{trans:eq:Schrodinger_solution}) that
    \[ \label{trans:eq:8}
        y_{+}(x,k) = e^{i\vp(x)}y_{1}(x,k) + e^{-i\vp(x)} y_{2}(x,k)\qq
        (x,k) \in \R_+ \ts \C \sm \{0\}.
    \]
    Since $\by$ is the Jost solution and $\sup \supp p \leq \g$, we have
    $$
        y_{1}(x,k) = e^{ikx},\qq y_{2}(x,k) = 0,\qq \vp(x) = 0,\qq x \geq \g.
    $$
    Substituting these identities in (\ref{trans:eq:8}), we get
    $$
        y_{+}(x,k) = e^{ikx},\qq x \geq \g.
    $$
    So that $y_+$ is the Jost solution of equation (\ref{trans:eq:Schrodinger_equation}).
\end{proof}

Recall for any $\a \in [0,\pi)$ the Jost function $\psi_{\a}$  and the
scattering matrix $S_{\a}$ of equation
(\ref{trans:eq:Schrodinger_equation}) were defined by
\[ \label{trans:eq:1}
    \begin{aligned}
        \psi_{\a}(k) &= y_+^{[1]}(0,k) \sin \a + k y_+(0,k) \cos \a,\qq k \in \C,\\
        S_{\a}(k) &= \frac{\ol{\psi_{\a} (k)}}{\psi_{\a} (k)},\qq k \in \R.
    \end{aligned}
\]
Recall also that for any $\a \in [0,\pi)$ we have introduced  the
mapping $\r^{\a}: \cQ \to \mathbf{s}$ such that $\r^{\a}(p,q) = (r_n)_{n \geq 1}$
is a sequence of resonances of equation (\ref{intro:eq:equation}) with boundary
condition (\ref{intro:eq:boundary_condition}), i.e. it is a sequence
of zeros of $\psi_{\a}(\cdot,p,q)$ in $\C_-$ counted with
multiplicities and arranged such that
$$
    \begin{cases}
        \abs{r_1} \leq \abs{r_2} \leq \ldots,\\
        \Re r_n \leq \Re r_{n+1}\qq \text{if $\abs{r_n} = \abs{r_{n+1}}$}.
    \end{cases}
$$

Using the connection between Jost solutions of Dirac and  energy
dependent Schr\"odinger equations, we get the following
representation.
\begin{lemma} \label{trans:lm:jost_functions}
    Let $v = \cT(p,q)$ for some $(p,q) \in \cQ$ and let $\vp(x) = \int_x^{+\iy} p(s) ds$.
    Then the following identities hold true
    $$
        \begin{aligned}
            \psi_{\a}(\cdot,p,q) = i e^{i\b} k \bpsi_{\d}(\cdot,v),\qq
            S_{\a}(\cdot,p,q) = - \bS_{\d}(\cdot,v),\qq
            \r^{\a}(p,q) = \br^{\d}(v),
        \end{aligned}
    $$
    where $\a,\d \in [0,\pi)$ and $\b \in \{0,\pi\}$ are such that
    \[ \label{trans:eq:boundary_parameters}
        \vp(0) + \a + \d= \frac{\pi}{2} \mod \pi,\qq
        \vp(0) + \a + \d= \frac{\pi}{2} + \b \mod 2\pi.
    \]
\end{lemma}
\begin{proof}
    Let $v = \cT(p,q)$ for some $(p,q) \in \cQ$ and let $\vp(x) = \int_x^{+\iy} p(s) ds$.
    Let $y_+$ be the Jost solution of equation (\ref{trans:eq:Schrodinger_equation}) and
    let $\by = \ma y_1 \\ y_2 \am$ be the Jost solution of equation (\ref{trans:eq:Dirac_equation}).

    Due to Lemma \ref{trans:lm:jost_solution}, we have
    \[ \label{trans:eq:10}
        \begin{cases}
            y_+(0,k) = e^{i\vp(0)} y_1(0,k) + e^{-i\vp(0)} y_2(0,k)\\
            k^{-1} y_+^{[1]}(0,k) = ie^{i\vp(0)} y_1(0,k) - ie^{-i\vp(0)} y_2(0,k).
        \end{cases}
    \]
    Substituting (\ref{trans:eq:10}) in (\ref{trans:eq:1}), we get
    \[ \label{trans:eq:eq2}
        \begin{aligned}
            \psi_{\a}(k,p,q) &= k(\sin \a (ie^{i\vp(0)} y_1(0,k) - ie^{-i\vp(0)} y_2(0,k))\\
            &+ \cos \a (e^{i\vp(0)} y_1(0,k) + e^{-i\vp(0)} y_2(0,k)))\\
            &= k(e^{i\vp(0)} y_1(0,k) (\cos \a + i\sin \a) + e^{-i\vp(0)} y_2(0,k)(\cos \a - i \sin \a))\\
            &= k(y_1(0,k) e^{i(\vp(0) + \a)} + y_2(0,k) e^{-i(\vp(0) + \a)}),\qq k \in \C.
        \end{aligned}
    \]
    Recall that
    $$
        \bpsi_{\d}(k,v) = y_{1}(0,k)e^{-i\d} - y_{2}(0,k)e^{i\d},\qq k \in \C.
    $$
    Due to (\ref{trans:eq:boundary_parameters}), we have
    $$
        \vp(0) + \a = \frac{\pi}{2} - \d + \b (\mod 2\pi).
    $$
    Substituting this formula in (\ref{trans:eq:eq2})
     and using $e^{i\b} = e^{-i\b}$, $e^{\frac{i\pi}{2}} = -e^{-\frac{i\pi}{2}}$, we get
    \[ \label{trans:eq:11}
        \begin{aligned}
            \psi_{\a}(k,p,q) &= k(y_1(0,k) e^{i(\frac{\pi}{2} - \d + \b)} + y_2(0,k) e^{-(\frac{\pi}{2} - \d + \b)}) \\
            &= k i e^{i\b}(y_1(0,k) e^{-i\d} - y_2(0,k) e^{i \d})
            = k i e^{i\b} \bpsi_{\d}(k,v),\qq k \in \C.
        \end{aligned}
    \]
    Hence the functions $\psi_{\a}(\cdot,q,p)$ and $\bpsi_{\d}(\cdot)$ have
    the same zeros in $\C_-$, which yields that $\r^{\a}(p,q) =
    \br^{\d}(v)$. Using (\ref{trans:eq:11}), we get
    $$
        S_{\a}(k,p,q) = \frac{\ol{\psi_{\a} (k,p,q)}}{\psi_{\a} (k,p,q)}
        = \frac{-k i e^{-i\b} \ol{\bpsi_{\d}(k,v)}}{k i e^{i\b} \bpsi_{\d}(k,v)}
        = -\frac{\ol{\bpsi_{\d}(k,v)}}{\bpsi_{\d}(k,v)} = -\bS_{\d}(k,v),\qq k \in \R.
    $$
\end{proof}
\begin{remark}
    Note that for any $\a \in [0,\pi)$ there exists a unique $\d \in
    [0,\pi)$ such that (\ref{trans:eq:boundary_parameters})  holds true
    and vice versa. And for any $\a,\d \in [0,\pi)$ there exists a
    unique $\b \in \{0,\pi\}$ such that
    (\ref{trans:eq:boundary_parameters}) holds true.
\end{remark}
At last we need a result about inverse scattering problem for
equation  (\ref{trans:eq:Schrodinger_equation}) with $(p,q) \in
\cQ_+$ (see Theorem 1 in \cite{HM20}). We introduce the following
class of the scattering matrices
\begin{definition}
    $\cS_+$ is the set of all functions $S: \R \to \C$ such that
    \begin{enumerate}
        \item $\abs{S(k)} = 1$ for all $k \in \R$;
        \item $\ind S = 0$;
        \item $S$ admit the following representation
        $$
            S(k) = e^{2i\a} + \int_{\R} F(s) e^{2iks} ds,\qq k \in \R,
        $$
        for some $\a \in [0,\pi)$ and $F \in L^2(\R,\C) \cap L^1(\R,\C)$.
    \end{enumerate}
\end{definition}

\begin{theorem} \label{trans:thm:scattering_inverse_problem}
The mapping $(p,q,\a) \mapsto S_{\a}(\cdot,p,q)$ from $\cQ_{+} \ts
[0,\pi)$ into $\cS_{+}$ is a bijection.
    Moreover, we have
    $$
        S_{\a}(k,p,q) = e^{-2i(\vp(0) + \a)} + \int_{\R} F(s) e^{2iks} ds,
    $$
    for some $F \in L^2(\R) \cap L^1(\R)$.
\end{theorem}

\section{Proof of the main theorems} \label{proof}

\begin{proof}[Proof of Theorem \ref{intro:thm:scattering_inverse_problem}]
    It follows from Theorem \ref{trans:thm:scattering_inverse_problem} that the mapping
    $(p,q,\a) \mapsto S_{\a}(\cdot,p,q)$ is a bijection between $\cQ_{+} \ts [0,\pi)$ and $\cS_{+}$.
    Now we only need to show that its restriction on $\cQ \ts [0,\pi)$ and $\cS$ is a surjection.
    Let $(p,q,\a) \in \cQ \ts [0,\pi)$ and let $v = \cT(p,q)$.
    Then it follows from Lemma \ref{trans:lm:jost_functions} that $S_{\a}(\cdot,p,q) = -\bS_{\d}(\cdot,v)$,
    where $\vp(0) + \a + \d = \frac{\pi}{2} \mod \pi$ and $\vp(x) = \int_x^{+\iy} p(s) ds$.
    Due to Lemma \ref{trans:lm:potential_mapping} we have $v \in \cX$ and then
    it follows from Theorem \ref{dirac:thm:inverse_jost} that $\bS_{\d}(\cdot,v) \in \cS_{\d}$.
    So that $S_{\a}(\cdot,p,q) \in \cS$. Moreover, it has the following representation
    $$
        S_{\a}(\cdot,p,q) = -\left(e^{2i\d} + \int_{-\g}^{+\iy} F(s) e^{2iks} ds\right).
    $$
    Substituting $\d = (\frac{\pi}{2} - \a - \vp(0)) \mod \pi$, we get
    $$
        S_{\a}(\cdot,p,q) = e^{-2i(\a + \vp(0))} - \int_{-\g}^{+\iy} F(s) e^{2iks} ds,
    $$
    which yields $S_{\a}(\cdot,p,q) \in \cS_{\b}$, where $\b = -(\a + \vp(0)) (\mod \pi)$.

    Now, let $S \in \cS$ and let $W = -S$. Then $W \in \cS_{\d}$ for
    some $\d \in [0,\pi)$. By Theorem \ref{dirac:thm:inverse_jost} there
    exists a unique $v \in \cX$ such that $\bS_{\d}(\cdot,v) = W$. Due to
    Lemma \ref{trans:lm:potential_mapping} there exists a unique $(p,q)
    = \cT^{-1}(v) \in \cQ$. And it follows from Lemma
    \ref{trans:lm:jost_functions} that
        $$
            S_{\a}(\cdot,p,q) = -\bS_{\d}(\cdot,v) = -W = S,
        $$
    where $\a = (\frac{\pi}{2} - \d - \vp(0)) \mod \pi$ and $\vp(x) =
    \int_x^{+\iy} p(s) ds$. So that  the mapping $(p,q,\a) \mapsto
    S_{\a}(\cdot,p,q)$ from $\cQ \ts [0,\pi)$ into $\cS$ is a bijection.
\end{proof}

\begin{proof}[Proof of Theorem \ref{intro:thm:resonances_inverse_problem}]
    i) Let $\a \in [0,\pi)$ and let $\z \in \cR$. We prove that there exists a unique
        $(p,q) \in \cQ$ such that $\r^{\a}(p,q) = \z$.

    By Theorem \ref{dirac:thm:inverse_resonances} there exists a unique
    $v_o \in \cX$ such that $\z = \br^{0}(v_o)$. Let $(p_o,q_o) =
    \cT^{-1}(v_o)$ and let $\a_o = \frac{\pi}{2} - \int_0^{\g} p(s)ds \mod \pi \in [0,\pi)$.
    By Lemma \ref{trans:lm:jost_functions}, we have $\r^{\a_o}(p_o,q_o) =
    \br^{0}(v_o) = \z$. So that $\r^{\a}(p,q) = \z$ for some $(p,q) \in
    \cQ$ if and only if $(p,q,\a) \in \Iso(p_o,q_o,\a_o)$.
    Thus, in order to prove the statement of the theorem, we only need to prove that
    there exists a unique $(p,q) \in \cQ$ such that $(p,q,\a) \in \Iso(p_o,q_o,\a_o)$.

    Due to Theorem \ref{intro:thm:isoresonance_potentials}, we have
    $$
        \Iso(p_o,q_o,\a_o) = \set*{(p_{\d},q_{\d},\a_{\d}) \given \d \in [0,\pi)},
    $$
    where $(p_{\d},q_{\d},\a_{\d})$ is given by (\ref{intro:eq:isoresonances_potentials}).
    In particular, we have
    $$
        \a_{\d} = \a_o + \vt_o(0) - \vt_{\d}(0) \mod \pi,
    $$
    where $\vt_{\d}$ is a unique solution of the initial value problem
    $$
        \begin{cases}
            \vt_{\d}' = - \Im v_o \cos 2 \vt_{\d} - \Re v_o \sin 2 \vt_{\d},\\
            \vt_{\d}(\g) = \d.
        \end{cases}
    $$
    Moreover, it follows from Theorem
    \ref{intro:thm:isoresonance_potentials} that $\vt_{\d}(0)$ is a
    strictly increasing function of $\d$ and $\vt_{\pi}(0) = \vt_o(0) +
    \pi$. Thus, for any $\a \in [0,\pi)$ there exists a unique $\d \in
    [0,\pi)$ such that $\a_{\d} = \a$. So that there exists a unique
    $(p_{\d},q_{\d})$ such that $(p_{\d},q_{\d},\a) \in
    \Iso(p_o,q_o,\a_o)$.

    ii) Let $(\z,\b) \in \cR \ts [0,\pi)$. Let $\d = \b - \frac{\pi}{2}
    \mod \pi$. By Theorem \ref{dirac:thm:inverse_resonances} there
    exists a unique $v \in \cX$ such that $\br^{\d}(v) = \z$. Due to
    Lemma \ref{trans:lm:potential_mapping} there exists a unique $(p,q)
    = \cT^{-1}(v) \in \cQ$. 
    Moreover, let $\a = \frac{\pi}{2} - \vp(0) - \d \mod \pi \in [0,\pi)$, where $\vp(x) = \int_x^{\g} p(s) ds$. 
    Then, by Lemma \ref{trans:lm:jost_functions}, we have
    $$
        S_{\a}(\cdot,p,q) = - \bS_{\d}(\cdot,v).
    $$
    Since $\bS_{\d}(\cdot,v) \in \cS_{\d}$ it follows that
    $$
        S_{\a}(\cdot,p,q) = -\bS_{\d}(\cdot,v) \in \cS_{\d + \frac{\pi}{2}} = \cS_{\b}.
    $$
\end{proof}

\begin{proof}[Proof of Theorem \ref{intro:thm:resonances_forbidden_domain}]
    Let $r_o \in \C_-$ be a resonance of equation (\ref{intro:eq:equation}) with
    the boundary condition (\ref{intro:eq:boundary_condition}) for some $(p,q,\a) \in \cQ \ts [0,\pi)$
    and let $\ve > 0$. Let $\r^{\a}(p,q) = (z_n)_{n \geq 1}$. Then $r_o = z_k$ for some $k \geq 1$.
    Let $v = \cT(p,q)$. Then, by Lemma
    \ref{trans:lm:jost_functions}, there exists a unique $\d \in
    [0,\pi)$ such that $\br^{\d}(v) = \r^{\a}(p,q) = (z_n)_{n \geq 1}$.
    Due to Theorem \ref{dirac:thm:forbidden_domain}, there exists a
    constant $C = C(v,\d,\ve) = C(p,q,\a,\ve)$ such that
    (\ref{intro:eq:forbidden_domain}) holds true for each $n \geq 1$ and then it holds true for $r_o$.
    Moreover, it follows from Theorem \ref{dirac:thm:forbidden_domain} that
    there exist a finitely many resonances in the strip given by
    (\ref{intro:eq:forbidden_domain_strip}).
\end{proof}

\begin{proof}[Proof of Theorem \ref{intro:thm:resonances_inverse_problem_schrodinger_operator}]
    Due to Theorem \ref{intro:thm:resonances_inverse_problem} we only need to prove that the
    mapping $\r^{o} : \cQ_o \to \cR_o$ is a surjection.

    Let $(p,q) \in \cQ$ and let $v = \cT(p,q)$. At first, we show that $p = 0$ if and only if $v \in \cX_{real}$.
    Let $p = 0$. Then we have $\vp = 0$ and it follows from (\ref{trans:eq:4}) that $\Im v = 0$.
    Let $v \in \cX_{real}$. Due to (\ref{trans:eq:potential_phase}) we have $\vp = 0$ and hence $p = 0$.

    Let $q \in \cQ_o$ and let $v = \cT(0,q)$. Since $v \in \cX_{real}$ it follows from Theorem
    \ref{dirac:thm:inverse_resonances_real_valued_potentials} that $\br^{\frac{\pi}{2}}(v) \in \cR_o$.
    Due to $p = 0$, we have $\vp = 0$. Using Lemma \ref{trans:lm:jost_functions}, we get
    $$
        \r^{o}(0,q) = \br^{\frac{\pi}{2}}(v) \in \cR_o.
    $$

    Let $\z \in \cR_o$. By Theorem \ref{dirac:thm:inverse_resonances_real_valued_potentials}
    there exists a unique $v \in \cX_{real}$ such that $\br^{\frac{\pi}{2}}(v) = \z$.
    Since $v \in \cX_{real}$ we have $\cT(0,q) = v$ for some $q \in \cQ_o$.
    Due to $p = 0$, we have $\vp = 0$. Using Lemma \ref{trans:lm:jost_functions}, we get
    $$
        \r^{o}(0,q) = \br^{\frac{\pi}{2}}(v) = \z.
    $$
\end{proof}

\begin{proof}[Proof of Theorem \ref{intro:thm:isoresonance_potentials}]
    Let $(p_o,q_o,\a_o) \in \cQ \ts [0,\pi)$ and let $v_o =
    \cT(p_o,q_o)$. Due to Lemma \ref{trans:lm:jost_functions} there
    exists a unique $\d_o \in [0,\pi)$ such that $\r^{\a_o}(p_o,q_o) =
    \br^{\d_o}(v_o)$ and
    \[ \label{proof:eq:9}
        \vp_o(0) + \a_o + \d_o = \frac{\pi}{2} \mod \pi,
    \]
    where $\vp_o(x) = \int_x^{\g} p_o(s) ds$.

    Let $(p,q,\a) \in \Iso(p_o,q_o,\a_o)$ and let $v = \cT(p,q)$. Due to
    Lemma \ref{trans:lm:jost_functions} there exists a unique $\b \in
    [0,\pi)$ such that $\r^{\a}(p,q) = \br^{\b}(v)$. Using $\r^{\a}(p,q)
    = \r^{\a_o}(p_o,q_o)$ and $\r^{\a_o}(p_o,q_o) = \br^{\d_o}(v_o)$, we
    get
    $$
        \br^{\b}(v) = \r^{\a}(p,q) = \r^{\a_o}(p_o,q_o) = \br^{\d_o}(v_o),
    $$
    which yields $(v,\b) \in \Iso(v_o, \d_o)$.

    On the other hand, let $v_{\d} = e^{2i\d} v_o$ and let
    $(p_{\d},q_{\d}) = \cT^{-1}(v_d)$ for some $\d \in [0,\pi)$. Due to
    Lemma \ref{trans:lm:jost_functions} there exists a unique $\a_{\d}
    \in [0,\pi)$ such that $\r^{\a_{\d}}(p_{\d},q_{\d}) = \br^{\d_o +
    \d}(v_{\d})$ and
    \[ \label{proof:eq:10}
        \vp_{\d}(0) + \a_{\d} + \d_o + \d = \frac{\pi}{2} \mod \pi,
    \]
    where $\vp_{\d}(x) = \int_x^{\g} p_{\d}(s) ds$. By Theorem
    \ref{dirac:thm:isoresonances_potentials}, for any $\d \in [0,\pi)$
    we have $\br^{\d_o + \d}(v_{\d}) = \br^{\d_o}(v_o)$, which yields
        $$
    \r^{\a_{\d}}(p_{\d},q_{\d}) = \br^{\d_o + \d}(v_{\d}) =
    \br^{\d_o}(v_o) = \r^{\a_o}(p_o,q_o).
    $$
    So that $(p_{\d},q_{\d},\a_{\d}) \in \Iso(p_o,q_o,\a_o)$. Thus, we obtain
    $$
        \Iso(p_o,q_o,\a_o) = \set*{(p_{\d},q_{\d},\a_{\d}) \given \d \in [0,\pi)},
    $$
    where $(p_{\d},q_{\d},\a_{\d})$ was defined above.

    Now, we obtain formulas for $p_{\d}$, $q_{\d}$ and $\a_{\d}$. By
    direct calculation, we get
    \[ \label{proof:eq:1}
        \Re v_{\d} = \Re v_o \cos 2\d - \Im v_o \sin 2\d,\qq \Im v_{\d} =
        \Re v_o \sin 2\d + \Im v_o \cos 2\d.
    \]
    Using (\ref{trans:eq:potential_phase}), we have
    \[ \label{proof:eq:2}
        \begin{cases}
            \vp_{\d}' = - \Im v_{\d} \cos 2 \vp_{\d} - \Re v_{\d} \sin 2 \vp_{\d},\\
            \vp_{\d}(\g) = 0.
        \end{cases}
    \]
    Substituting (\ref{proof:eq:1}) in (\ref{proof:eq:2}), we get
    \[ \label{proof:eq:3}
        \begin{aligned}
            \vp_{\d}' &= - (\Re v_o \sin 2\d +
            \Im v_o \cos 2\d) \cos 2 \vp_{\d} - (\Re v_o \cos 2\d - \Im v_o \sin 2\d) \sin 2 \vp_{\d}\\
            &= -\Im v_o (\cos 2\d \cos 2 \vp_{\d} - \sin 2\d \sin 2 \vp_{\d})
            -\Re v_o (\sin 2\d \cos 2 \vp_{\d} + \cos 2\d \sin 2 \vp_{\d})\\
            &= -\Im v_o \cos 2(\d + \vp_{\d}) - \Re v_o \sin 2(\d + \vp_{\d}).
        \end{aligned}
    \]
    Let $\vt_{\d} = \d + \vp_{\d}$. Then, using (\ref{proof:eq:2}) and
    (\ref{proof:eq:3}), we get
    \[ \label{proof:eq:5}
        \begin{cases}
            \vt_{\d}' = - \Im v_o \cos 2 \vt_{\d} - \Re v_o \sin 2 \vt_{\d},\\
            \vt_{\d}(\g) = \d.
        \end{cases}
    \]
    Using $\vt_{\d} = \d + \vp_{\d}$ and (\ref{proof:eq:9}), we get
    $$
        \d_o = \frac{\pi}{2} - \a_o - \vt_o(0) \mod \pi
    $$
    Substituting this formula in (\ref{proof:eq:10}), we obtain
    $$
        \a_{\d} = \a_o + \vt_o(0) - \vt_{\d}(0) \mod \pi.
    $$
    Recall that $p_{\d} = - \vp_{\d}'$ and then it follows from (\ref{proof:eq:3}) that
    $$
        p_{\d} = \Im v_o \cos 2 \vt_{\d} + \Re v_o \sin 2 \vt_{\d}.
    $$
    Let $q_{\d} = u_{\d}' + u_{\d}^2$. Using
    (\ref{trans:eq:potentials_transform}) and (\ref{proof:eq:1}), we get
    $$
        \begin{aligned}
            u_{\d} &= - \Re v_{\d} \cos 2 \vp_{\d} + \Im v_{\d} \sin 2 \vp_{\d} =\\
            &= - (\Re v_o \cos 2\d - \Im v_o \sin 2\d) \cos 2 \vp_{\d} +
            (\Re v_o \sin 2\d + \Im v_o \cos 2\d) \sin 2 \vp_{\d} =\\
            &= - \Re v_o (\cos 2 \d \cos 2 \vp_{\d} - \sin 2\d \sin 2\vp_{\d})
            + \Im v_o (\sin 2 \d \cos 2 \vp_{\d} + \cos 2\d \sin 2\vp_{\d})=\\
            &= - \Re v_o \cos 2(\d + \vp_{\d}) + \Im v_o \sin 2(\d + \vp_{\d}).
        \end{aligned}
    $$
    Substituting $\vt_{\d} = \d + \vp_{\d}$ in this equation, we get
    \[ \label{proof:eq:7}
        u_{\d} = - \Re v_o \cos 2 \vt_{\d} + \Im v_o \sin 2 \vt_{\d}.
    \]

    At last, we study the motion of $\vt_{\d}(0)$. Note that equation
    (\ref{proof:eq:5}) has a unique solution and then $\vt_{\d + \pi} =
    \vt_{\d} + \pi$ for any $\d \in \R$. In particular, $\vt_{\pi} =
    \vt_o + \pi$. Let $\dot y = \frac{\partial y}{\partial \d}$.
    Differentiating (\ref{proof:eq:5}) by $\d$ and using
    (\ref{proof:eq:7}), we get
    $$
        \begin{cases}
            \dot \vt_{\d}' = 2 u_{\d} \dot \vt_{\d},\\
            \dot \vt_{\d}(\g) = 1.
        \end{cases}
    $$
    Solving this first-order linear differential equation, we get
    $$
        \dot \vt_{\d}(x) = e^{-2 \int_x^{\g} u_{\d}(s) ds}.
    $$
    Thus, we have $\dot \vt_{\d}(0) > 0$, which yields that
    $\vt_{\d}(0)$ is a strictly increasing function of $\d$.
\end{proof}

\begin{proof}[Proof of Theorem \ref{intro:thm:isoresonances_scattering_matrices}]
    Let $(p_o,q_o,\a_o) \in \cQ \ts [0,\pi)$ and let $v_o =
    \cT(p_o,q_o)$. It follows from the proof of Theorem
    \ref{intro:thm:isoresonance_potentials} that $\cT(p_{\d},q_{\d}) =
    v_{\d} = e^{2i\d} v_o = e^{2i\d} \cT(p_o,q_o)$. Using Lemma
    \ref{trans:lm:jost_functions} and (\ref{proof:eq:10}),
    we get
    \[ \label{proof:eq:4}
        S_{\a_{\d}}(\cdot,p_{\d},q_{\d}) = - \bS_{\d_o + \d}(\cdot, v_{\d})
    \]
    for any $\d \in [0,\pi)$. Since $v_{\d} = e^{2i\d} v_o$, it follows
    from Theorem \ref{dirac:thm:isoresonances_potentials} that
    \[ \label{proof:eq:8}
        \bS_{\d_o + \d}(\cdot,v_{\d}) = e^{2i\d} \bS_{\d_o}(\cdot,v_o)
    \]
    for any $\d \in [0,\pi)$. Combining (\ref{proof:eq:4}) and (\ref{proof:eq:8}), we get
    $$
        S_{\a_{\d}}(\cdot,p_{\d},q_{\d}) = - \bS_{\d_o + \d}(\cdot, v_{\d}) =
        - e^{2i\d} \bS_{\d_o}(\cdot,v_o) = e^{2i\d} S_{\a_o}(\cdot,p_o,q_o).
    $$
\end{proof}

\begin{proof}[Proof of Theorem \ref{intro:thm:reducing_of_scattering_problem}]
    Let $\a \in [0,\pi)$. Firstly, we construct the mappings $\phi_{\a}$ and $\xi_{\a}$.
    Let $(p,q) \in \cQ$. Let $v_o = \cT(p,q)$ and let $\vt_{\d}$, $\d \in [0,\pi)$, be a solution
    of initial value problem (\ref{intro:eq:isoresonances_potentials_ivp}).
    Let also $p_{\d}$, $q_{\d}$ and
    $\a_{\d}$ be given by (\ref{intro:eq:isoresonances_potentials}).
    Due to Theorem \ref{intro:thm:isoresonance_potentials}
    $\vt_{\d}(0)$ is a strictly increasing function of $\d$, and then there exists a unique $\d' \in [0,\pi)$
    such that
    \[ \label{proof:eq:12}
        \a_{\d'} = \a + \vt_o(0) - \vt_{\d'}(0) = 0 \mod \pi.
    \]
    Due to Theorem \ref{intro:thm:isoresonance_potentials} it follows that
    $(p_{\d'},q_{\d'},0) \in \Iso(p,q,\a)$, which yields $\xi_{\a}(p,q) = (p_{\d'},q_{\d'})$.
    Let $\vp_{\d}(x) = \int_x^{+\iy} p_{\d}(s) ds$.
    Recall that in the proof of Theorem \ref{intro:thm:isoresonance_potentials} $\vt_{\d}$ was defined such that
    we have $\vt_{\d}(0) = \d + \vp_{\d}(0)$. Substituting this formula in definition of $\phi_{\a}$
    and using (\ref{proof:eq:12}), we get
    $$
        \phi_{\a}(p,q) = \a + \vp_{o}(0) - \vp_{\d'}(0) = \a + \vt_o(0) - \vt_{\d'}(0) + \d' = \d'.
    $$
    Now, it follows from Theorem \ref{intro:thm:isoresonances_scattering_matrices} that
    \[ \label{proof:eq:11}
        S_{\a}(\cdot,p,q) = e^{-2i\d'} S_{\a_{\d'}}(\cdot,p_{\d'},q_{\d'}) = e^{-2i \phi_{\a}(p,q)} S_{0}(\cdot, \xi_{\a}(p,q)).
    \]

    We prove that the mapping $\xi_{\a}$ is a bijection. Due to Theorem \ref{intro:thm:isoresonance_potentials}
    there exist a unique $\d \in [0,\pi)$ such that $(p_{\d},q_{\d},\a_{\d}) \in \Iso(p,q,\a)$ and
    $\a_{\d} = 0$, where $(p_{\d},q_{\d},\a_{\d})$ is given by (\ref{intro:eq:isoresonances_potentials})
    and (\ref{intro:eq:isoresonances_potentials_ivp}). Thus, the mapping $\xi_{\a}$ is an injection.

    Now we prove that the mapping $\xi_{\a}$ is an surjection. Let $(\tilde{p},\tilde{q}) \in \cQ$. Due to
    Theorem \ref{intro:thm:isoresonance_potentials} there exists a unique $(p,q) \in \cQ$ such that
    $(p,q,\a) \in \Iso(\tilde{p},\tilde{q},0)$ and it follows from the definition of the
    isoresonances set that $(\tilde{p},\tilde{q},0) \in \Iso(p,q,\a)$. So that we have
    $\xi_{\a}(p,q) = (\tilde{p},\tilde{q})$.
\end{proof}

\end{document}